\documentclass[a4paper,11pt]{article}

\usepackage[english]{babel}
\usepackage[utf8]{inputenc} 
\usepackage[T1]{fontenc}    
\usepackage{lmodern}
\usepackage{hyperref}       
\usepackage{url}            
\usepackage{booktabs}       
\usepackage{amsfonts}       
\usepackage{nicefrac}       
\usepackage{microtype}      
\usepackage{enumitem}

\usepackage{graphicx}
\usepackage[titletoc,title]{appendix}
\usepackage{amsmath,amsthm,amssymb,mathtools}
\usepackage{braket}
\newtheorem{theorem}{Theorem}[section]
\newtheorem*{theorem*}{Theorem}

\newtheorem{proposition}[theorem]{Proposition}
\newtheorem*{proposition*}{Proposition}
\newtheorem{lemma}[theorem]{Lemma}
\newtheorem*{lemma*}{Lemma}

\newtheorem*{conjecture*}{Conjecture}

\newtheorem*{fact*}{Fact}

\newtheorem*{hypothesis*}{Hypothesis}

\newtheorem{itheorem}[theorem]{Informal Theorem}

\newtheorem*{claim*}{Claim}

\theoremstyle{definition}
\newtheorem{definition}[theorem]{Definition}

\theoremstyle{remark}
\newtheorem{remark}[theorem]{Remark}
\newtheorem*{remark*}{Remark}

\newtheorem*{observation*}{Observation}

\newcommand{\R}{\mathbb{R}}

\newcommand{\Rplus}{\R^+}

\newcommand{\calC}{\mathcal{C}}

\newcommand{\calG}{\mathcal{G}}
\newcommand{\calN}{\mathcal{N}}

\newcommand{\calX}{\mathcal{X}}

\newcommand{\poly}{\mathrm{poly}}

\newcommand{\Abs}[1]{\left\lvert#1\right\rvert}

\newcommand{\Bigset}[1]{\Big\{#1\Big\}}
\newcommand{\norm}[1]{\lVert #1 \rVert}

\newcommand{\iprod}[1]{\langle#1\rangle}

\newcommand{\Esymb}{\mathbb{E}}
\newcommand{\Psymb}{\mathbb{P}}

\DeclareMathOperator*{\E}{\Esymb}

\DeclareMathOperator*{\ProbOp}{\Psymb}

\renewcommand{\Pr}{\ProbOp}

\newcommand{\tmu}{\tilde{\mu}}

\newcommand{\eps}{\varepsilon}
\renewcommand{\epsilon}{\varepsilon}

\newcommand{\wmin}{w_{\text{min}}}

\newcommand{\ehat}{\widehat{e}}

\newcommand{\eii}{\widehat{e}_i}

\newif\ifnotes\notesfalse

\ifnotes
\usepackage{color}
\definecolor{mygrey}{gray}{0.50}
\newcommand{\notename}[2]{{\textcolor{mygrey}{\footnotesize{\bf (#1:} {#2}{\bf ) }}}}

\else

\newcommand{\notename}[2]{{}}

\fi

\newcommand{\pnote}[1]{{\notename{Pranjal}{#1}}}
\newcommand{\anote}[1]{{\notename{Aravindan}{#1}}}

\DeclarePairedDelimiter\inner{\langle}{\rangle}
\DeclarePairedDelimiter\abs{\lvert}{\rvert}

\DeclarePairedDelimiter\parens{(}{)}

\newcommand{\ybar}{\bar{y}}
\newcommand{\xbar}{\bar{x}}

\usepackage{tikz}
\usepackage{multirow}
\usepackage[noend]{algpseudocode}

\usepackage[a4paper,top=3cm,bottom=2cm,left=3cm,right=3cm,marginparwidth=1.75cm]{geometry}

\usepackage{amsmath}
\usepackage{graphicx}
\usepackage[colorinlistoftodos]{todonotes}

\title{Clustering Semi-Random Mixtures of Gaussians}
\author{Pranjal Awasthi\thanks{Department of Computer Science, Rutgers University.}
 \\{\tt pranjal.awasthi@rutgers.edu} \and
Aravindan Vijayaraghavan\thanks{
  Department of Electrical Engineering and Computer Science,
  Northwestern University. Supported by the National Science Foundation (NSF) under Grant No.~CCF-1652491 and CCF-1637585.} \\
  {\tt aravindv@northwestern.edu}
}
\date{}

\begin{document}
\maketitle

\begin{abstract}
\pnote{Changes}
Gaussian mixture models (GMM) are the most widely used statistical model for the $k$-means clustering problem and form a popular framework for clustering in machine learning and data analysis. In this paper, we propose a natural semi-random model for $k$-means clustering that generalizes the Gaussian mixture model, and that we believe will be useful in identifying robust algorithms.
In our model, a semi-random adversary is allowed to make arbitrary ``monotone'' or helpful changes to the data generated from the Gaussian mixture model. 

Our first contribution is a polynomial time algorithm that provably recovers the ground-truth up to small classification error w.h.p., assuming certain separation between the components. Perhaps surprisingly, the algorithm we analyze is the popular Lloyd's algorithm for $k$-means clustering that is the method-of-choice in practice. 
Our second result complements the upper bound by giving a nearly matching information-theoretic lower bound on the number of misclassified points incurred by any $k$-means clustering algorithm on the semi-random model.



\end{abstract}

\section{Introduction}\label{sec:intro}

\pnote{Changes.}
Clustering is a ubiquitous task in machine learning and data mining for partitioning a data set into groups of similar points. 
The $k$-means clustering problem is arguably the most well-studied problem in machine learning. However, designing provably optimal $k$-means clustering algorithms is a challenging task as the $k$-means clustering objective is NP-hard to optimize~\cite{SW} (in fact, it is also NP-hard to find near-optimal solutions~\cite{awasthi2015hardness,lee2017improved}). 
A popular approach to cope with this intractability is to study average-case models for the $k$-means problem.  
The most widely used such statistical model for clustering is the {\it Gaussian Mixture Model~(GMM)}, that has a long and rich history~\cite{Tei61,Pea94,Das99,AK01,VW04,DS07,BV08,MV10,BS10,KK10}. 

In this model there are $k$ clusters, and the points from cluster $i$ are generated from a Gaussian in $d$ dimensions with mean $\mu_i \in \R^d$, and covariance matrix $\Sigma_i \in \R^{d \times d}$ with spectral norm $\norm{\Sigma_i} \le \sigma^2$. Each of the $N$ points in the instance is now generated independently at random, and is drawn from the $i$th component with probability $w_i \in [0,1]$ ($w_1, w_2, \dots, w_k$ are also called mixing weights). 
If the means of the underlying Gaussians are separated enough, the ground truth clustering is well defined\footnote{A separation of $\norm{\mu_i-\mu_j}_2 \ge \Omega(\sigma \sqrt{\log (Nk)})$ for $i \ne j \in [k]$ suffices w.h.p.}. The algorithmic task is to recover the ground truth clustering for any data set generated from such a model (note that the parameters of the Gaussians, mixing weights and the cluster memberships of the points are unknown). 

Starting from the seminal work of Dasgupta~\cite{Das99}, there have been a variety of algorithms to provably cluster data from a GMM model. Algorithms based on PCA and distance-based clustering~\cite{AK01,VW04,AM05,KSV08} provably recover the clustering when there is adequate separation between  every pair of components (parameters). 
Other algorithmic approaches include the method-of-moments~\cite{KMV10,MV10,BS10}, and algebraic methods based on tensor decompositions~\cite{HK12,GVX14,BCMV,ABGRV14,GHK}. (Please see Section~\ref{sec:relatedwork} for a more detailed comparison of the guarantees).  

\anote{Changes.}
On the other hand, the method-of-choice in practice are iterative algorithms like the Lloyd's algorithm (also called  $k$-means algorithm)~\cite{lloyd1982least} and the $k$-means++ algorithm of ~\cite{AV07} (Lloyd's algorithm initialized with centers from distance-based sampling). In the absence of good worst-case guarantees, a compelling direction is to use {\em beyond-worst-case} paradigms like average-case analysis to provide provable guarantees. Polynomial time guarantees for recovering $k$-means optimal clustering by the Lloyd's algorithm and $k$-means++ are known when the points are drawn from a GMM model under sufficient separation conditions~\cite{DS07, KK10,AS12}.


Although the study of Gaussian mixture models has been very fruitful in designing a variety of efficient algorithms, real world data rarely satisfies such strong distributional assumptions. 
Hence, our choice of algorithm should be informed not only by its computational efficiency but also by its robustness to errors and model misspecification. 
As a first step, we need theoretical frameworks that can distinguish between algorithms that are tailored towards a specific probabilistic model and algorithms robust to modeling assumptions. In this paper we initiate such a study in the context of clustering, by studying a natural {\em semi-random model} that generalizes the GMM model and also captures robustness to certain adversarial dependencies in the data. 


Semi-random models involve a set of adversarial choices in addition to the random choices of the probabilistic model, while generating the instance. These models have been successfully applied to study the design of robust algorithms for various optimization problems~\cite{BS92,FK99,MS10,KMM,MMV12,MMV14} (see Section~\ref{sec:relatedwork}) 
In a typical semi-random model, there is a ``planted'' or ``ground-truth'' solution, and an instance is first generated according to a simple probabilistic model. An adversary is then allowed to make ``monotone'' or helpful changes to the instance that only make the planted solution more pronounced. 
For instance, in the semi-random model of Feige and Kilian~\cite{FK99} for graph partitioning, the adversary is allowed to arbitrarily add extra edges within each cluster or delete edges between different clusters of the planted partitioning.  
These adversarial choices only make the planted partition more prominent; however, the choices can be dependent and thwart algorithms that rely on the excessive independence or strong but unrealistic structural properties of these instances. 

Hence, the study of semi-random models helps us understand and identify robust algorithms. 
Our motivation for studying semi-random models for clustering is two-fold: a) design algorithms that are robust to strong distributional data assumptions, and b) explain the empirical success of simple heuristics such as the Lloyd's algorithm.

\paragraph{Semi-random mixtures of Gaussians}

In an ideal clustering instance, each point $x$ in the $i$th cluster is significantly closer to the mean $\mu_i$ than to any other mean $\mu_j$ for $j \ne i$ (for a general instance, in the optimal solution, $\norm{x-\mu_i}_2 - \norm{x-\mu_j}_2 \le 0~ \forall j \ne i$). Moving each point in $C_i$ toward its own mean $\mu_i$ only increases this gap between the distance to its mean and to any other mean. Hence, this perturbation corresponds to a monotone perturbation that only make this planted clustering even better. In our semi-random model, the points are first drawn from a mixture of Gaussians (this is the planted clustering). The adversary is then allowed to move each point in the $i$th cluster closer to its mean $\mu_i$. This allows the points to be even better clustered around their respective means, however these perturbations are allowed to have arbitrary dependencies.  
 We now formally define the semi-random model.
\begin{definition}[Semi-random GMM model]\label{def:srmodel}
Given a set of parameters $\mu_1, \mu_2, \dots, \mu_k \in \R^d$ and $\sigma \in \Rplus$, a clustering instance 
$\calX$ on $N$ points is generated as follows. 
\begin{enumerate}
  \item Adversary chooses an arbitrary partition $\calC=(C_1,C_2, \dots, C_k)$ of $\set{1,\dots,N}$ and let $N_i=|C_i|$ for all $i \in [k]$.  
  \item For each $i \in [k]$ and each $t \in C_i$, $y^{(t)} \in \R^d$ is generated independently at random according to a Gaussian with mean $\mu_i$ and covariance $\Sigma_i$ with $\norm{\Sigma} \le \sigma$ i.e., variance at most $\sigma^2$ in each direction. 
  \item The adversary then moves each point $y^{(t)}$ towards the mean of its component by an arbitrary amount i.e., for each $i \in [k], t \in C_i$, the adversary picks $x^{(t)}$ arbitrarily in $\set{\mu_i+\lambda(y^{(t)}-\mu_i): \lambda \in [0,1]}$. (Note that these choices can be correlated arbitrarily.)   
 \end{enumerate}
The instance is $\calX=\set{x^{(t)}: t \in [N]}$ and is parameterized by $(\mu_1, \dots, \mu_k, \sigma)$ with the planted clustering $C_1, \dots, C_k$. We will denote by $\wmin=\min_{i \in [k]} N_i/ N$.
\end{definition}

Data generated by mixtures of high-dimensional Gaussians have certain properties that are often not exhibited by real-world instances. High-dimensional Gaussians have strong concentration properties; for example, all the points generated from a high-dimensional Gaussian are concentrated at a reasonably far distance from the mean (they are $\approx \sqrt{d} \sigma$ far away w.h.p.). In many real-world datasets on the other hand, clusters in the ground-truth often contain dense ``cores'' that are close to the mean. Our semi-random model admits such instances by allowing points in a cluster to move arbitrarily close to the mean. 


\paragraph{Our Results.}
Our first result studies the Lloyd's algorithm on the semi-random GMM model and gives an upper bound on the clustering error achieved by the Lloyd's algorithm with the initialization procedure used in~\cite{KK10}.  
\begin{itheorem}\label{informalthm:upperbound}
Consider any semi-random instance $\calX$ with $N$ points generated by the semi-random GMM model (Definition~\ref{def:srmodel}) with planted clustering $C_1, \dots, C_k$ and parameters $\mu_1, \dots, \mu_k, \sigma^2$ satisfying
\begin{equation}
\forall i \ne j \in [k],~\norm{\mu_i - \mu_j}_2 > \Delta \sigma, ~\text{ where } \Delta\ge c_0 \sqrt{\min\set{k,d} \log N}, \nonumber
\end{equation}
and $N \ge k^2 d^2/\wmin^2$. There is polynomial time algorithm based on the Lloyd's iterative algorithm that recovers the cluster memberships of all but $\tilde{O}(kd/\Delta^4)$ points.   
\end{itheorem}
The $\tilde{O}$ in the above statement hides a $\log(\log N/\Delta^4)$ and $\log(d/\Delta^2)$  factor. 
Please see Theorem~\ref{thm:upperbound} for a formal statement. 
\pnote{Changes.} 
Furthermore, we show that in the above result the Lloyd's iterations can be initialized using the popular $k$-mean++ algorithm that uses $D^2$-sampling~\cite{AV07}. The most closely related to our work is that of~\cite{KK10} and \cite{AS12} who provided deterministic data conditions under which the Lloyd's algorithm converges to the optimal clustering. Along these lines, our work provides further theoretical justification for the enormous empirical success that the Lloyd's algorithm enjoys.

\pnote{Changes.}
It is also worth noting that in spite of being robust to semi-random perturbations, the separation requirement of $\sigma \sqrt{k \log N}$ in our upper bound matches the separation requirement in the best guarantees \cite{AS12} for Lloyd's algorithm even in the absence of any semi-random errors or perturbations (see Section~\ref{sec:relatedwork} for a comparison) \footnote{We note that for clustering GMMs, the work of Brubaker and Vempala~\cite{BV08} give a qualitatively different separation condition that does not depend on the maximum variance, and can model Gaussian mixtures that look like ``parallel pancakes''. However this separation condition is incomparable to \cite{AS12}, because of the potentially worse dependence on $k$.}. We also remark that while the algorithm recovers a clustering of the given data that is very close to the planted clustering, this does not necessarily estimate the means of the original Gaussian components up to inverse polynomial accuracy (in fact the centers of the planted clustering after the semi-random perturbation may be $\Omega(\sigma)$ far from the original means). This differs from the recent body of work on parameter estimation in the presence of some adversarial noise (please refer to Section~\ref{sec:relatedwork} for a comparison).  


While the monotone changes allowed in the semi-random model should only make the clustering task easier, our next result shows that the error achieved by the Lloyd's algorithm is in fact near optimal. More specifically, we provide a lower bound on the number of points that will be misclassified by any $k$-means optimal solution for the instance.

\begin{itheorem}\label{informalthm:lowerbound}
Given any $N$ (that is sufficiently large polynomial in $d,k$) and $\Delta$ such that  $\sqrt{\log N} \le \Delta \le d/(4 \log d)$, there exists an instance $\calX$ on $N$ points in $d$ dimensions generated from the semi-random GMM model~\ref{def:srmodel} with parameters $\mu_1, \dots, \mu_k, \sigma^2$, and planted clustering $C_1, \dots, C_k$ having separation $\forall i \ne j \in [k], \norm{\mu_i - \mu_j}_2 \ge \Delta \sigma$ s.t. any optimal $k$-means clustering solution $C'_1, C'_2, \dots, C'_k$ of $\calX$ misclassifies at least $\Omega(k d/\Delta^4)$ points with high probability.   
\end{itheorem}
The above lower bound also holds when the semi-random (monotone) perturbations are applied to points generated from a mixture of $k$ spherical Gaussians each with covariance $\sigma^2 I$ and weight $1/k$. Further, the lower bound holds not just for the optimal $k$-means solution, but also for any ``locally optimal'' clustering solution.  Please see Theorem~\ref{thm:lowerbound} for a formal statement. These two results together show that the Lloyd's algorithm essentially recovers the planted clustering up to the optimal error possible for any $k$-means clustering based algorithm. 


Unlike algorithmic results for other semi-random models, an appealing aspect of our algorithmic result is that it gives provable robust guarantees in the semi-random model for a simple, popular algorithm that is used in practice (Lloyd's algorithm). 
Further, other approaches for clustering like distance-based clustering, method-of-moments and tensor decompositions seem inherently non-robust to these semi-random perturbations (see Section~\ref{sec:relatedwork} for details). This robustness of the Lloyd's algorithm suggests an explanation for its widely documented empirical success across different application domains.



\paragraph{Considerations in the choice of the Semi-random GMM model.}

Here we briefly discuss different semi-random models, and considerations involved in favoring Definition~\ref{def:srmodel}. Another semi-random model that comes to mind is one that can move each point closer to the mean of its own cluster (closer just in terms of distance, regardless of direction). Intuitively this seems appealing since this improves the cost of the planted clustering. However, in this model the optimal $k$-means clustering of the perturbed instance can be vastly different from the planted solution. This is because one can move many points $x$ in cluster $C_i$ in such a way that $x$ becomes closer to a different mean rather than $\mu_i$. For high dimensional Gaussians it is easy to see that the distance of each point to its own mean will be on the order of $(\sqrt{d}+2\sqrt{\log N})\sigma$. Hence, in our regime of interest, the inter mean separation of $\sqrt{k\log N} \sigma$  could be much smaller than the radius of any cluster (when $d \gg k$). 
Consider an adversary that moves a large fraction of the points in a given cluster to the mean of another cluster. While the distance of these points to their cluster mean has only decreased from roughly $(\sqrt{d}+2\sqrt{\log N})\sigma$ to around $\sqrt{k \log N}\sigma$, these points now become closer to the mean of a different cluster! In the semi-random GMM model on the other hand, the adversary is only allowed to move the point $x$ along the direction of $x-\mu$; hence, each point $x$ becomes closer to its own mean than to the means of other clusters. Our results show that in such a model, the optimal clustering solution can change by at most $\tilde{O}(d/\Delta^4)$ points.

\anote{Added the "hence, each point becomes close...." line.}

\paragraph{Challenges in the Semi-random GMM model and Overview of Techniques.}

Lloyd's algorithm has been analyzed in the context of clustering mixtures of Gaussians~\cite{KK10, AS12}. Any variant of the Lloyd's algorithm consists of two steps --- an initialization stage where a set of $k$ initial centers are computed, and the iterative algorithm which successively improves the clustering in each step. Kumar and Kannan~\cite{KK10} considered a variant of the Lloyd's method where the initialization is given by using PCA along with a $O(1)$ factor approximation to the $k$-means optimization problem. The improved analysis of this algorithm in~\cite{AS12} leads to state of the art results that perfectly recovers all the clusters when the separation is of the order $\sqrt{k \log N} \sigma$.


We analyze the variant of Lloyd's algorithm that was introduced by Kumar and Kannan~\cite{KK10}. However, there are several challenges in extending the analysis of \cite{AS12} to the semi-random setting. While the semi-random perturbations in the model only move points in a cluster $C_i$ closer to the mean $\mu_i$, these perturbations can be co-ordinated in a way that can move the empirical mean of the cluster significantly. For instance, Lemma~\ref{lem:movingcenter} gives a simple semi-random perturbation to the points in $C_i$ that moves the empirical mean of the points in $C_i$ to $\tmu_i$ s.t. $\tmu_i \approx \mu_i + \Omega(\sigma) \ehat$, for any desired direction $\ehat$. This shift in the empirical means may now cause some of the points in cluster $C_i$ to become closer to $\tmu_j$ (in particular points that have a relatively large projection onto $\ehat$) and vice-versa.  In fact, the lower bound instance in Theorem~\ref{thm:lowerbound} is constructed by applying such a semi-random perturbation given by Lemma~\ref{lem:movingcenter} to the points in a cluster, along a carefully picked direction  so that $m=\Omega(d/\Delta^4)$ points are misclassified per cluster. 


The main algorithmic contribution of the paper is an analysis of the Lloyd's iterative algorithm when the points come from the semi-random GMM model. The key is to understand the number of points that can be misclassified in an intermediate step of the Lloyd's iteration. We show in Lemma~\ref{lem:badpoints} that if in the current iteration of the Lloyd's algorithm, each of the current estimates of the means $\mu'_i$ is within $\tau \sigma$ from $\mu_i$, then the number of misclassified points by the current iteration of Lloyd's iteration is at most $\tilde{O}(kd \tau^2/ \Delta^4)$. This relies crucially on Lemma~\ref{lem:baddirections} which upper bounds the number of points $x$ in a cluster $C_i$ s.t. $(x-\mu_i)$ has a large inner product along any (potentially bad) direction $\ehat$. 

The effect of these bad points has to be carefully accounted for when analyzing both stages of the algorithm -- the initialization phase, and the iterative algorithm. 
Proposition~\ref{prop:initialization-strong} argues about the closeness of the initial centers to the true means. As in~\cite{KK10}, these initial centers are obtained via  a boosting technique that first maps the points to an expanded feature space and then uses the ($k$-SVD + $k$-means approximation) to get initial centers. When using this approach for semi-random data one needs to carefully argue about how the set of bad points behave in the expanded feature space. This is done in Lemmas~\ref{lem:boost-mean-separation} and ~\ref{lem:boosting-spectral-norm}.
Given the initial centers, it is not hard to see that the analysis of ~\cite{AS12} can be carried out to argue about the improvements made in the Lloyd's iterative step; however, this leads to a bound that is sub-optimal by the factor of $O(k^2)$. Instead, we perform a much finer analysis for the semi-random model to control the effect of the bad points and achieve nearly optimal error bounds. This is done in Lemma~\ref{lem:lloyds-iterate}.

\subsection{Related Work} \label{sec:relatedwork}

There has been a long line of algorithmic results on Gaussian mixture models starting from ~\cite{Tei61,Tei67, Pea94}. These results fall into two broad categories: 
{\em (1) Clustering algorithms}, which aim to recover the component/cluster memberships of the points and {\em (2) Parameter estimation}, where the goal is to estimate the parameters of the Gaussian components. When the components of the mixture are sufficiently well-separated, i.e., $\norm{\mu_i - \mu_j}_2 \ge \sigma \sqrt{\log (Nk)}$, then the Gaussians do not overlap w.h.p., and then the two tasks become equivalent w.h.p. We now review the different algorithms that have been designed for these two tasks, and comment on their robustness to semi-random perturbations.      

\paragraph{Clustering Algorithms.} The first polynomial time algorithmic guarantees were given by Dasgupta~\cite{Das99}, who showed how to cluster a mixture of $k$ Gaussians with identical covariance matrices when the separation between the cluster means is of the order $\Omega(\sigma \sqrt{d} \text{polylog}(N) )$, 
where $\sigma$ denotes the maximum variance of any cluster along any direction\footnote{The $\text{polylog}(N)$ term involves a dependence of either $(\log N)^{1/4}$ or $(\log N)^{1/2}$. }. 
Distance-based clustering algorithms that are based on strong distance-concentration properties of high-dimensional Gaussians improved the separation requirement between means $\mu_i$ and $\mu_j$ to be $\Omega(d^{1/4} \text{polylog}(N)) (\sigma_i+\sigma_j)$ ~\cite{AK01,DS07}, where $\sigma_i$ denotes the maximum variance of points in cluster $i$ along any direction. Vempala and Wang~\cite{VW04} and subsequent results~\cite{KSV08,AM05} used PCA to project down to $k$ dimensions (when $k \le d$), and then used the above distance-based algorithms to get state-of-the-art guarantees for many settings: for spherical Gaussians a separation of roughly $\norm{\mu_i - \mu_j}_2 \ge (\sigma_i+\sigma_j) \min\set{k,d}^{1/4} \text{polylog}(N)$ suffices~\cite{VW04}. For non-spherical Gaussians, a separation of $\norm{\mu_i-\mu_j}_2 \ge (\sigma_i+\sigma_j) k^{3/2} \sqrt{\log N}$ is known to suffice~\cite{AM05,KSV08}. Brubaker and Vempala~\cite{BV08} gave a qualitative improvement on the separation requirement for non-spherical Gaussians by having a dependence only on the variance along the direction of the line joining the respective means, as opposed to the maximum variance along any direction.  

Recent work has also focused on provable guarantees for heuristics such as the Lloyd's algorithm for clustering mixtures of Gaussians~\cite{KK10, AS12}. Iterative algorithms like the Lloyd's algorithm (also called  $k$-means algorithm)~\cite{lloyd1982least} and its variants like $k$-means++~\cite{AV07} are the method-of-choice for clustering in practice. 
The best known guarantee~\cite{AS12} along these lines requires a separation of order $\sigma \sqrt{k \log N}$ between any pair of means, where $\sigma$ is the maximum variance among all clusters along any direction. To summarize, for a mixture of $k$ Gaussians in $d$ dimensions with variance of each cluster being bounded by $\sigma^2$ in every direction, the state-of-the-art guarantees require a separation of roughly $\sigma \min\set{k,d}^{1/4} \text{polylog}(N)$ between the means of any two components~\cite{VW04} for spherical Gaussians, while a separation of $\sigma \sqrt{\min\set{k,d} \log N}$ is known to suffice for non-spherical Gaussians~\cite{AS12}. 


The techniques in many of the above works rely on strong distance concentration properties of high-dimensional Gaussians. For instance, the arguments of~\cite{AK01,VW04} that obtain a separation of order $\min\set{k^{1/4}, d^{1/4}}$ crucially rely on the tight concentration of the squared distance around $\sigma^2(d \pm c\sqrt{d})$, between any pair of points in the same cluster. 
These arguments do not seem to carry over to the semi-random model. Brubaker~\cite{Bru} gave a robust algorithm for clustering a mixture of Gaussians when at most $o(1/k)$ fraction of the points are corrupted arbitrarily. However, it is unclear if the arguments can be modified to work under the semi-random model, since the perturbations can potentially affect all the points in the instance. On the other hand, our results show that the Lloyd's algorithm of Kumar and Kannan~\cite{KK10} is robust to these semi-random perturbations.


\paragraph{Parameter Estimation.} A different approach is to design algorithms that estimate the parameters of the underlying Gaussian mixture model, and then assuming the means are well separated, accurate clustering can be performed. A very influential line of work focuses on the method-of-moments~\cite{KMV10,MV10,BS10} to learn the parameters of the model when the number of clusters $k=O(1)$. Moment methods (necessarily) require running time (and sample complexity) of roughly $d^{O(k^2)}$, but do not assume any explicit separation between the components of the mixture. 
Recent work~\cite{HK13,BCV,GVX14,BCMV,ABGRV14,GHK} uses uniqueness of tensor decompositions (of order $3$ and above) to implement the method of moments and give polynomial time algorithms assuming the means are sufficiently high dimensional, and do not lie in certain degenerate configurations~\cite{HK12,GVX14,BCMV,ABGRV14,GHK}.

Algorithmic approaches based on method-of-moments and tensor decompositions rely heavily on the exact parametric form of the Gaussian distribution and the exact algebraic expressions to express various moments of the distribution in terms of the parameters. These algebraic methods can be easily foiled by a monotone adversary, since the adversary can perturb any subset to alter the moments significantly (for example, even the first moment, i.e., the mean of a cluster, can change by $\Omega(\sigma)$). 

Recent work has also focused on provable guarantees for heuristics such as Maximum Likelihood estimation and the Expectation Maximization (EM) algorithm for parameter estimation~\cite{DS07,BWY, Hsuetal16, DTZ16}.
Very recently, \cite{RV17} considered other iterative algorithms for parameter estimation, and studied the optimal order of separation required for parameter estimation. However, we are not aware of any existing analysis that shows that these iterative algorithms for parameter estimation are robust to modeling errors.

\anote{Changes.}
Another recent line of exciting work concerns designing robust high-dimensional estimators of the mean and covariance of a single Gaussian (and mixtures of $k$ Gaussians) when an $\eps=\Omega_k(1)$ fraction of the points are adversarially corrupted~\cite{DKKLMS,LRV16,CSV17}. However, these results and similar results on agnostic learning do not necessarily recover the ground-truth clustering. Further, they typically assume that only a $o(1/k)$ fraction of the points are corrupted, while potentially all the points could be perturbed in the semi-random model. On the other hand, our work does not necessarily give guarantees for estimating the means of the original Gaussians (in fact the centers given by the planted clustering in the semi-random instance can be $\Omega(\sigma)$ far from the original means). Hence, our semi-random model is incomparable to the model of robustness considered in these works.   

\paragraph{Semi-random models for other optimization problems.} 
There has been a long line of work on the study of semi-random models for various optimization problems. Blum and Spencer~\cite{BS92} initiated the study of semi-random models, and studied the problem of graph coloring. Feige and Kilian~\cite{FK99} considered semi-random models involving monotone adversaries for various problems including graph partitioning, independent set and clique. Makarychev et al.~\cite{MMV12, MMV14} designed algorithms for more general semi-random models for various graph partitioning problems. The work of~\cite{MPW15} studied the power of monotone adversaries in the context of community detection (stochastic block models), while \cite{MMVSBM} considered the robustness of community detection to monotone adversaries and different kinds of errors and model misspecification. Semi-random models have also been studied for correlation clustering~\cite{MS10,MMVCC}, noisy sorting~\cite{MMVfas} and coloring~\cite{DF16}.

\section{Preliminaries and Semi-random model}\label{sec:model}
\anote{Maybe introduce $k$-means problem, Voronoi partition and locally optimal solution for $k$-means? 
Also introduce GMM model more formally?}
We first formally define the Gaussian mixture model.
\begin{definition} (Gaussian Mixture Model). A Gaussian mixture model with $k$ components is defined by the parameters $(\mu_1, \mu_2, \ldots \mu_k, \Sigma_1, \ldots, \Sigma_k, w_1, \ldots, w_k)$. Here $\mu_i \in \mathbb{R}^d$ is the mean for component $i$ and $\Sigma_i \in \mathbb{S}^{d}_{+}$ is the corresponding $d \times d$ covariance matrix. $w_i \in [0,1]$ is the mixing weight and we have that $\sum_{i=1}^k w_i=1$. An instance $\calX = \set{x^{(1)},\dots, x^{(N)}}$ from the mixture is generated as follows: for each $t \in [N]$, sample a component $i \in [k]$ independently at random with probability $w_i$. Given the component, sample $x^{(t)}$ from $\mathcal{N}(\mu_i, \Sigma_i)$. The $N$ points can be naturally partitioned into $k$ clusters $C_1, \ldots, C_k$ where cluster $C_i$ corresponds to the points that are sampled from component $i$. We will refer to this as the {\em planted clustering} or {\em ground truth clustering}.
\end{definition}
Clustering data from a mixture of Gaussians is a natural average-case model for the $k$-means clustering problem. Specifically, if the means of a Gaussian mixture model are well separated, then with high probability, the ground truth clustering of an instance sampled from the model corresponds to the $k$-means optimal clustering.
\begin{definition} ($k$-means clustering). Given an instance $\calX=\set{x^{(1)},\dots, x^{(N)}}$ of $N$ points in $\mathbb{R}^d$, the $k$-means problems is to find $k$ points $\mu_1, \ldots, \mu_k$ such as to minimize $\sum_{t \in [N]} \min_{i \in [k]} \|x^{(t)}-\mu_i\|^2$.
\end{definition}
The optimal means or centers $\mu_1, \ldots, \mu_k$ naturally define a clustering of the data where each point is assigned to its closest cluster. A key property of the $k$-means objective is that the optimal solution induces a locally optimal clustering.
\begin{definition} (Locally Optimal Clustering). A clustering $C_1, \ldots, C_k$ of $N$ data points in $\mathbb{R}^d$ is locally optimal if for each $i \in [k], x^{(t)} \in C_i$, and $j \neq i$ we have that $\|x^{(t)}-\mu(C_i)\| \leq \|x^{(t)}-\mu_j\|$. Here $\mu(C_i)$ is the average of the points in $C_i$.
\end{definition}

Hence, given the optimal $k$-means clustering, the optimal centers can be recovered by simply computing the average of each cluster. This is the underlying principle behind the popular Lloyd's algorithm~\cite{lloyd1982least} for $k$-means clustering. The algorithm starts with a choice of initial centers. It then repeatedly computes new centers to be the average of the clusters induced by the current centers. Hence the algorithm converges to a locally optimal clustering. Although popular in practice, the worst case performance of Lloyd's algorithm can be arbitrarily bad~\cite{arthur2005worst}. 
\pnote{Changes.}The choice of initial centers is very important in the success of the Lloyd's algorithm. We show that our theoretical guarantees hold when the initialization is done via the popular $k$-means++ algorithm~\cite{AV07}. There also exist more sophisticated constant factor approximation algorithms for the $k$-means problem~\cite{kanungo2002local,AhmadianNSW16} that can be used for seeding in our framework.

While the clustering $C_1, C_2, \dots, C_k$ typically represents a partition of the index set $[N]$, we will sometimes abuse notation and use $C_i$ to also denote the set of points in $\calX$ that correspond to these indices in $C_i$. Finally, many of the statements are probabilistic in nature depending on the randomness in the semi-random model. In the following section, w.h.p. will refer to a probability of at least $1-o(1)$ (say $1-1/\text{poly}(N)$), unless specified otherwise.   

\subsection{Properties of Semi-random Gaussians}
\label{sec:prelim-sr}
In this section we state and prove properties of semi-random mixtures that will be used throughout the analysis in the subsequent sections. 
We first start with a couple of simple lemmas that follow directly from the corresponding lemmas about high dimensional Gaussians.

\begin{lemma}\label{lem:sr:length}
Consider any semi-random instance $\calX=\set{x^{(1)},\dots, x^{(N)}}$ with parameters $\mu_1, \dots, \mu_k, \sigma^2$ and clusters $C_1, \dots, C_k$. Then with high probability we have
\begin{equation}\label{eq:sr:length}
\forall i \in [k], ~\forall \ell \in C_i,~~ \norm{x^{(t)}-\mu_i }_2 \le \sigma(\sqrt{d}+2\sqrt{\log N})  .
\end{equation}
\end{lemma}
\begin{proof}
Let $y^{(t)}$ denote the point generated in the semi-random model in step 2 (Definition~\ref{def:srmodel}) before the semi-random perturbation was applied. Let $\xbar^{(t)}=x-\mu_i, ~\ybar^{(t)}=y-\mu_i$ where $t \in C_i$. We have 
$$\forall i \in [k], \forall t \in C_i,~~ \norm{\xbar^{(t)}}_2 \le  \norm{\ybar^{(t)}}_2 \le \sigma(\sqrt{d}+2\sqrt{\log N}),$$
from Lemma~\ref{lem:amphelp:length}. 
\end{proof}

\begin{lemma}\label{lem:sr:innerprod}
Consider any semi-random instance $\calX=\set{x^{(1)},\dots, x^{(N)}}$ with parameters $\mu_1, \dots, \mu_k, \sigma^2$ and clusters $C_1, \dots, C_k$, and let $u$ be a fixed unit vector in $\R^d$. Then with probability at least $(1-1/(N^3))$ we have
\begin{equation}\label{eq:sr:innerprod}
\forall i \in [k], t \in C_i, ~~ \abs{\iprod{x^{(\ell)}-\mu_i,u}} < 3\sigma\sqrt{\log N} .
\end{equation}
\end{lemma}
\begin{proof}
Let $y^{(t)}$ denote the point generated in the semi-random model in step 2 (Definition~\ref{def:srmodel}) before the semi-random perturbation was applied. Let $\xbar^{(t)}=x-\mu_i, ~\ybar^{(t)}=y-\mu_i$ where $t \in C_i$. 

Consider the sample $t \in C_i$. Let $\Sigma_i$ be the covariance matrix of $i$th Gaussian component; hence $\norm{\Sigma_i} \le \sigma$. The projection $\iprod{\ybar^{(t)},u}$ is a Gaussian with mean $0$ and variance $u^T \Sigma_i u \le \sigma^2$. From Lemma~\ref{lem:gaussian1d} 
$$\Pr\Big[|\iprod{\xbar^{(t)}, u}| \ge 3 \sigma \sqrt{\log N} \Big] \le \Pr\Big[|\iprod{\ybar^{(t)}, u}| \ge 3 \sigma \sqrt{\log N} \Big] \le \exp(-4 \log N) \le N^{-4}.$$
Hence from a union bound over all $N$ samples, the lemma follows. 
\end{proof}

The above lemma immediately implies the following lemma after a union bound over the $k^2<N^2$ directions given by the unit vectors along $(\mu_i-\mu_j)$ directions.

\begin{lemma}\label{lem:sr:innerprod:means}
Consider any semi-random instance $\calX=\set{x^{(1)},\dots, x^{(N)}}$ with parameters $\mu_1, \dots, \mu_k, \sigma^2$ and clusters $C_1, \dots, C_k$
. Then with high probability we have
\begin{equation}\label{eq:sr:innerprod:means}
\forall i \in [k], t \in C_i, ~~ \Abs{\inner[\Big]{x^{(\ell)}-\mu_i,\frac{\mu_i - \mu_j}{\norm{\mu_i-\mu_j}_2}}} < 3\sigma\sqrt{\log N} .
\end{equation}
\end{lemma}


We next state a lemma about how far the mean of the points in a component of a semi-random GMM can move away from the true parameters.
\begin{lemma}
\label{lem:semirandom-mean}
Consider any semi-random instance $\mathcal{X}$ with $N$ points generated with parameters $\mu_1, \ldots, \mu_k, C_1, \ldots , C_k$ such that $N_i  \geq 4(d+\log(\frac k \delta))$ for all $i \in [k]$. Then with probability at least $1-\delta$ we have that 
\begin{equation}\label{eq:sr:distance:means}
\forall i \in [k], \norm{\frac 1 {|C_i|}\sum_{x \in C_i} x - \mu_i }_2 \leq 2\sigma.
\end{equation}
\end{lemma}
\begin{proof}
For each point $x \in C_i$ in the semi-random GMM, let $y_x$ be the original point in the GMM that is modified to produce $x$. Then, we know that $x-\mu_i = \lambda_x(y_x-\mu_i)$ where $\lambda_x \in [0,1]$. Hence, $\frac 1 {|C_i|} \sum_{x \in C_i} (x-\mu_i) = \frac 1 {|C_i|} A_i D v$, where $A_i$ is the matrix with columns as $(y_x-\mu_i)$ for $x \in C_i$, $D$ is a diagonal matrix with values $\lambda_x$, and $v$ is a unit length vector in the direction of $\frac 1 {|C_i|} \sum_{x \in C_i} (x-\mu_i)$. Then, we have that $\|\frac 1 {|C_i|} \sum_{x \in C_i} x - \mu_i\| = \|\frac 1 {|C_i|} A_i D v\| \leq \frac 1 {|C_i|} \|A\| \leq 2\sigma$~(from \ref{lem:Gaussian-variance-concentration}).
\end{proof}

The next lemma argues about the variance of component $i$ around $\mu_i$ in a semi-random GMM.
\begin{lemma}
\label{lem:semirandom-variance}
Consider any semi-random instance $\mathcal{X}$ with $N$ points generated with parameters $\mu_1, \ldots, \mu_k, C_1, \ldots , C_k$ such that $N_i  \geq 4(d+\log(\frac k \delta))$ for all $i \in [k]$. Then with probability at least $1-\delta$ we have that 
\begin{equation}\label{eq:sr:variance}
\forall i \in [k], \max_{v:\|v\|=1} \frac 1 {|C_i|}\sum_{x \in C_i}\abs{\iprod{x - \mu_i, v}}^2 \leq 4\sigma^2.
\end{equation}
\end{lemma}
\begin{proof}
Exactly as in the proof of Lemma~\ref{lem:semirandom-mean}, we can write $\max_{v:\|v\|=1} \frac 1 {|C_i|}\sum_{x \in C_i}\abs{\iprod{x - \mu_i,v}}^2 = \max_{v:\|v\|=1} \frac 1 {|C_i|}\sum_{x \in C_i}\abs{\lambda^2_x\iprod{y_x - \mu_i,v}}^2 \leq \max_{v:\|v\|=1} \frac 1 {|C_i|}\sum_{x \in C_i}\abs{\iprod{y_x - \mu_i, v}}^2$. Furthermore, since $y_x$ are points from a Gaussian we know that with probability at least $1-\delta$, for all $i \in [k]$, $\max_{v:\|v\|=1} \|\frac 1 {|C_i|} \sum_{x \in C_i} \abs{\iprod{y_x - \mu_i, v}}^2 \leq 4\sigma^2$. Hence, the claim follows. 
\end{proof}

We would also need to argue about the mean of a large subset of points from a component of a semi-random GMM.
\begin{lemma}
\label{lem:semi-random-mean-of-subset}
Consider any semi-random instance $\mathcal{X}$ with $N$ points generated with parameters $\mu_1, \ldots, \mu_k$ and  planted clustering $C_1, \ldots , C_k$ such that $N_i  \geq 16(d+\log(\frac k \delta))$ for all $i \in [k]$. Let $G_i \subseteq C_i$ be such that $|G_i| \geq (1-\epsilon)|C_i|$ where $\epsilon < \frac 1 2$. Then, with probability at least $1-\delta$, we have that 
\begin{equation}\label{eq:sr:mean:subset}
\forall i \in [k], \|\mu(G_i) - \mu_i\| \leq (4 + \frac{2}{\sqrt{1-\epsilon}})\sigma.
\end{equation}
\end{lemma}
\begin{proof}
Let $C_i$ be the set of points in component $i$ and let $\nu_i$ be the mean of the points in $C_i$. Notice that from Lemma~\ref{lem:semirandom-mean} and the fact that the perturbation is semi-random, we have that with probability at least $1-\frac \delta 2$, $\|\nu_i - \mu_i\| \leq 2\sigma$. Also, because the component is a semi-random perturbation of a Gaussian, we have from Lemma~\ref{lem:semirandom-variance} that $\frac 1 {|C_i|} \max_{v: \|v\|=1} \sum_{x \in C_i} [\iprod{x-\nu_i, v}^2] \leq 4\sigma^2$ with probability at least $1-\frac \delta 2$. 

Hence, with probability at least $1-\delta$ we have that $\|\mu(G_i) - \mu_i\| \leq \|\nu_i - \mu_i\| + \|\mu(G_i)-\nu_i\| \leq 4\sigma + \|\mu(G_i)-\nu_i\|$. To bound the second term notice that $\|\mu(G_i) - \nu_i\| = |(\frac 1 {|G_i|} \sum_{x \in G_i} \iprod{x-\nu_i,\hat{u}}|$, where $\hat{u}$ is a unit vector in the direction of $(\mu(G_i)-\nu_i)$. Using Cauchy-Schwarz inequality, this is at most $\frac 1 {\sqrt{|G_i|}} \sqrt{\sum_{x \in C_i}\iprod{x-\nu_i,\hat{u}}^2} \leq \frac{2\sigma}{\sqrt{1-\epsilon}}$. Combining the two bounds gives us the result.
\end{proof}
Finally, we argue about the variance of the entire data matrix of a semi-random GMM.
\begin{lemma}
\label{lem:semirandom-all-variance}
Consider any semi-random instance $\mathcal{X}$ with $N$ points generated with parameters $\mu_1, \ldots, \mu_k, C_1, \ldots , C_k$ such that $N_i  \geq 4(d+\log(\frac k \delta))$ for all $i \in [k]$. Let $A \in \R^{d \times N}$ be the matrix of data points and let $M \in \mathbb{R}^{d \times N}$ be the matrix composed of the means of the corresponding clusters. Then, with probability at least $1-\delta$, we have that 
\begin{equation}\label{eq:sr:variance:all}
\|A-M\| \leq 4\sigma \sqrt{N}.
\end{equation}
\end{lemma} 
\begin{proof}
Let $M^*$ be the matrix of true means corresponding to the cluster memberships. We can write $\|A-M\| \leq \|A-M^*\| + \|M^*-M\|$. Using Lemma~\ref{lem:semirandom-mean}, we know that with probability at least $1-\frac \delta 2$, $\max_i \|M^*_i-M_i\| \leq 2\sigma$. Hence, $\|M^*-M\| \leq 2\sigma \sqrt{N}$. Furthermore, $\|A-M^*\|^2 = \max_{v:\|v\|=1} \sum_{i} \sum_{x \in C_i} |(x-\mu_i)\cdot v|^2$. From Lemma~\ref{lem:semirandom-variance}, with probability at least $1 - \frac \delta 2$, we can bound the sum by at most $4\sigma^2 N$. Hence, $\| A-M^* \| \leq 2\sigma \sqrt{N}$. Combining the two bounds we get the claim. 
\end{proof}



The following lemma is crucial in analyzing the performance of the Lloyd's algorithm. We would like to upper bound the inner product $\abs{\iprod{x^{(\ell)}-\mu_i,\ehat}} < \lambda \sigma$ for every direction $\ehat$ and sample $\ell \in [N]$, but this is impossible since $\ehat$ can be aligned along $x^{(\ell)}-\mu_i$. The following lemma however upper bounds the total number of points in the dataset that can have a large projection of $\lambda$ (or above) onto any direction $\ehat$ by at most $\tilde{O}(d/\lambda^2)$. This involves a union bound over a net of all possible directions $\ehat$. 

\anote{This lemma is the main place where we lose the extra log factor in the number of points. Could this be improved using some Talagrand-style chaining?}
\begin{lemma}[Points in Bad Directions]\label{lem:baddirections}
Consider any semi-random instance $\calX=\set{x^{(1)}, \dots, x^{(N)}}$ with $N$ points having parameters $\mu_1, \dots, \mu_k, \sigma^2$ and planted clustering $C_1, \dots, C_k$, and suppose $\forall i \in [k], \ell \in C_i, ~ \xbar^{(\ell)}=x^\ell-\mu_i$. Then there exists a universal constant $c>0$ s.t. for any $\lambda >100\sqrt{\log N}$, with probability at least $1- 2^{-d}$, we have that 
\begin{equation} \label{eq:baddirections}
\forall \ehat \in \R^d \text{ s.t. } \norm{\ehat}_2=1, ~~ \Abs{\set{\ell \in [N]: \abs{\iprod{\xbar^{(\ell)}, \ehat}}> \lambda \sigma}} \le  \frac{c d }{\lambda^2}\cdot \max\Bigset{1,\log\parens[\big]{\tfrac{3(\sqrt{d}+2\sqrt{\log N})}{\lambda}}}.
\end{equation}
\end{lemma}
\anote{Explain some nuances about the statement here. We can't hope to say this for every point. However, this lemma shows that it is true except for a few bad points. }

\begin{proof}
Set $\eta:=\min\set{\lambda/(2\sqrt{d}+2\sqrt{\log N}),\tfrac{1}{2}}$ and $m:= 512 d \log(3/\eta)/\lambda^2$.
Consider an $\eta$-net $\calN \subset \set{u: \norm{u}_2=1}$ over unit vectors in $\R^d$. Hence 
$$\forall u \in \R^{d}: \norm{u}_2=1,~\exists v \in \calN \text{ s.t. } \norm{u-v}_2 \le \eta \text{ and } ~|\calN| \le \parens[\Big]{\frac{2+\eta}{\eta}}^{d} \le \exp\parens[\big]{d \log(3/\eta)}.$$ 

Further, since $|\iprod{\xbar, \ehat}| > \lambda$ and $\calN$ is an $\eta$-net, there exists some unit vector $u=u(\ehat) \in \calN$
\begin{align} 
\abs{\iprod{\xbar,u}}&>\abs{\iprod{\xbar, \ehat}+\iprod{\xbar, \ehat-u}} \geq \sigma \lambda - \norm{\xbar}_2 \norm{\ehat - u}_2 \ge \sigma \parens[\big]{\lambda - \eta(\sqrt{d}+2\sqrt{\log N}) } \ge \frac{\lambda}{2}, \label{eq:badprojection}
\end{align}
for our choice of $\eta$. 
Consider a fixed $x \in \set{x^{(1)}, \dots, x^{(N)}}$ and a fixed direction $u \in \calN$. Since the variance of $y$ is at most $\sigma^2$ we have
$$\Pr\Big[ \abs{\iprod{\xbar, u}} > \lambda \sigma/2 \Big] \le \Pr\Big[ \abs{\iprod{\ybar, u}} > \lambda \sigma/2 \Big] \le \exp\parens[\big]{- \lambda^2/8}. $$
The probability that $m$ points in $\set{x^{(1)}, \dots, x^{(N)}}$ satisfy \eqref{eq:badprojection} for a fixed direction $u$ is at most ${N \choose m} \cdot \exp(-m\gamma^2 /2)$. Let $E$ represent the bad event that there exists a direction in $\calN$ such that more that $m$ points satisfy the bad event given by \eqref{eq:badprojection}. The probability of $E$ is at most
\begin{align*}
\Pr[E] &\le |\calN|\cdot {N \choose m} \exp(-m\lambda^2/8) \le \exp\parens[\Big]{d \log(3/\eta) + m \log N - \frac{m\lambda^2}{8}}\\
&\le \exp\parens[\Big]{-d\log(1/\eta)} \le \eta^d,
\end{align*}
since for our choice of parameters $\lambda^2> 32 \log N$, and $m\lambda^2 \ge 32 d \log(3/\eta)$.

\end{proof}

\section{Upper Bounds for Semi-random GMMs}

\anote{Wherever you see a $(\sqrt{d}+\sqrt{\log N})$ term, change it to $\sqrt{d}+2\sqrt{\log N}$ --- this arises from the concentration bound for the length of a Gaussian vector. }
\anote{Change $\tau<\Delta/4$ to $\tau<\Delta/(24)$; changed in theorem statement, but need to change it in proofs.}
In this section we prove the following theorem that provides algorithmic guarantees for the Lloyd's algorithm with appropriate initialization, under the semi-random model for mixtures of Gaussians in Definition~\ref{def:srmodel}.

\anote{Check if $\min\set{k,d}$ is enough in separation.}
\begin{theorem}\label{thm:upperbound}
There exists a universal constant $c_0, c_1>0$ such that the following holds. There exists a polynomial time algorithm that for any semi-random instance $\calX$ on $N$ points with planted clustering $C_1, \dots, C_k$ generated by the semi-random GMM model (Definition~\ref{def:srmodel}) with parameters $\mu_1, \dots, \mu_k, \sigma^2$ s.t. 
\begin{equation} \label{eq:separation}
\forall i \ne j \in [k],~\norm{\mu_i - \mu_j}_2 > \Delta \sigma, ~\text{ where } \Delta>c_0 \sqrt{\min\set{k,d} \log N},
\end{equation}
and $N \ge k^2 d^2/\wmin^2$
finds w.h.p. a clustering $C'_1, C'_2, \dots, C'_k$ such that 
$$ \min_{\pi \in \text{Perm}_k} \sum_{i=1}^k \Abs{C_{\pi(i)} \triangle C'_i}\le \frac{c_1 k d}{\Delta^4}\cdot \max\Bigset{1,\log\parens[\big]{\tfrac{3(\sqrt{d}+2\sqrt{\log N})}{\Delta^2}}}.$$ 
\end{theorem}
\anote{Need to figure out exact dependence on $N$, this is based on initialization and Lloyd's iterate.}
 
\anote{Make some remark about the $\log$ factor. And compare to lower bound which is off by maybe this log factor}

In Section~\ref{sec:lower-bound} we show that the above error bound is close to the information theoretically optimal bound~(up to the logarithmic factor). The Lloyd's algorithm as described in Figure~\ref{ALG:Lloyds} consists of two stages, the initialization stage and an iterative improvement stage. 
\begin{figure}[hb]
\begin{center}
\fbox{\parbox{0.98\textwidth}{
\begin{enumerate}
\item Let $A$ be the $N \times d$ data matrix with rows $A_i$ for $i \in [N]$. Use $A$ to compute initial centers $\mu^{(1)}_{0}, \mu^{(2)}_{0}, \ldots \mu^{(k)}_{0}$ as detailed in Proposition~\ref{prop:initialization-strong}.
\item Use these $k$-centers to seed a series of Lloyd-type iterations. That is, for $r = 1,2, \ldots$ do:
\begin{itemize}
\item Set $Z_i$ be the set of points for which the closest center among $\mu^{(1)}_{r-1}, \mu^{(2)}_{r-1}, \ldots, \mu^{(k)}_{r-1}$ is $\mu^{(i)}_{r-1}$.
\item Set $\mu^{(i)}_{r} \leftarrow \frac 1 {|Z_i|} \sum_{A_j \in Z_i} A_j$.
\end{itemize}
\end{enumerate}
}}
\end{center}
\caption{\label{ALG:Lloyds} Lloyd's Algorithm}
\end{figure}

\anote{Give a description of the algorithm here.}
The initialization follows the same scheme as proposed by Kumar and Kannan in \cite{KK10}. The initialization algorithm first performs a $k$-SVD of the data matrix followed by running the $k$-means++ algorithm~\cite{AV07} that uses $D^2$-sampling to compute seed centers. One can also use any constant factor approximation algorithm for $k$-means clustering in the projected space to obtain the initial centers~\cite{kanungo2002local,AhmadianNSW16}. This approach works for clusters that are nearly balanced in size. However, when the cluster sizes are arbitrary, an appropriate transformation of the data is performed first that amplifies the separation between the centers. Following this transformation, the ($k$-SVD $+$ $k$-means++) is used to get the initial centers. The formal guarantee of the initialization procedure is encapsulated in the following proposition, whose proof is given in Section~\ref{sec:initialization}.    

The main algorithmic contribution of this paper is an analysis of the Lloyd's algorithm when the points come from the semi-random GMM model. For the rest of the analysis we will assume that the instance $\calX$ generated from the semi-random GMM model satisfies \eqref{eq:sr:length} to \eqref{eq:baddirections}. These eight equations are shown to hold w.h.p. in Section~\ref{sec:prelim-sr} for instances generated from the model. Our analysis will in fact hold for any deterministic data set satisfying these equations. This helps to gracefully argue about performing many iterations of Lloyd's on the same data set without the need to draw fresh samples at each step. 

\begin{proposition}
\label{prop:initialization-strong}
In the above notation for any $\delta>0$, suppose we are given an instance $\calX$ on $N$ points satisfying \eqref{eq:sr:length}-\eqref{eq:baddirections} such that $|C_i| \geq \Omega(d+\log(\frac k \delta))$ and assume that $\Delta \geq {125}\sqrt{{\min{\set{k,d}}\log N}}$. Then after the initialization step, for every $\mu_i$ there exists $\mu'_i$ such that $\|\mu_i - \mu'_i\| \leq \tau \sigma$, where $\tau < \Delta/24$.
\end{proposition}

\anote{Is separation dependence $\min\set{k,d}$?}

The analysis of the Lloyd's iterations crucially relies on the following lemma that upper bounds the number of misclassified points when the current Lloyd's iterative is relatively close to the true means. 

\begin{lemma}[Projection condition] \label{lem:badpoints}
In the above notation, consider an instance $\calX$ satisfying \eqref{eq:sr:length}-\eqref{eq:baddirections} and \eqref{eq:separation}
and suppose we are given $\mu'_1,\dots,\mu'_k$ satisfying $\forall j \in [k],~\norm{\mu'_j - \mu_j}_2 \le \tau \sigma$ and $\tau< \Delta/24$. 
Then there exists a set $Z \subset \calX$ such that for any $i \in [k]$ we have  
$$\forall x \in C_i \cap (\calX \setminus Z),~~\norm{x-\mu'_i}_2^2 \le \min_{j \ne i}\norm{x-\mu'_j}_2^2 ~~\text{ where } |Z|=O\parens[\Big]{\frac{d \tau^2 }{\Delta^4}\cdot \max\Bigset{1,\log\parens[\big]{\tfrac{3\tau(\sqrt{d}+2\sqrt{\log N})}{\Delta^2}}}}.$$ 
\end{lemma}


The following lemma quantifies the improvement in each step of the Lloyd's algorithm. The proof uses Lemma~\ref{lem:badpoints} along with properties of semi-random Gaussians. 

\begin{lemma}
\label{lem:lloyds-iterate}
In the above notation, suppose we are given an instance $\calX$ on $N$ points with $w_i N \geq \frac{d\sqrt{d}}{4 \log(d)}$ for all $i$ satisfying \eqref{eq:sr:length}-\eqref{eq:baddirections}. Furthermore, suppose we are given centers $\mu'_1, \dots, \mu'_k$ such that $\|\mu'_i - \mu_i\| \leq \tau \sigma,~ \forall i \in [k]$ where $\tau < \Delta/24$. Then the centers $\mu''_1,\dots, \mu''_k$ obtained after one Lloyd's update satisfy $\|\mu''_i - \mu_i\| \leq \max((6 + \frac \tau 4)\sigma, \frac \tau 2 \sigma))$ for all $i \in [k]$.
\end{lemma}


\anote{Write down all the deterministic conditions from the preliminaries that we use. }

We now present the proof of Theorem~\ref{thm:upperbound}. 
\begin{proof}[Proof of Theorem~\ref{thm:upperbound}]
Firstly, the eight deterministic conditions \eqref{eq:sr:length}-\eqref{eq:baddirections} are shown to hold for instance $\calX$ w.h.p. in Section~\ref{sec:prelim-sr}.
The proof follows in a straightforward manner by combining Proposition~\ref{prop:initialization-strong}, Lemma~\ref{lem:lloyds-iterate} and Lemma~\ref{lem:badpoints}. Proposition~\ref{prop:initialization-strong} shows that $\norm{\mu^{(0)}_i -\mu_i}_2 \le \Delta/(24)$ for all $i \in [k]$. Applying Lemma~\ref{lem:lloyds-iterate}, we have that after $T=O(\log \Delta)$ iterations we get $\norm{\mu^{(T)}_i - \mu_i}_2 \le 8 \sigma$ for all $i \in [k]$ w.h.p. Finally using Lemma~\ref{lem:badpoints} with $\tau=1$, the theorem follows.    
\end{proof}

\subsection{Analyzing Lloyd's Algorithm}

We now analyze each iteration of the Lloyd's algorithm and show that we make progress in each step by misclassifying fewer points with successive iterations. As a first step we begin with the proof of Lemma~\ref{lem:badpoints}. 
\begin{proof}[Proof of Lemma~\ref{lem:badpoints}]
Set $m:= 512 d \log(3/\eta)\tau^2 /\Delta^4$ where $\eta=\min\set{\Delta^2/(\tau(2\sqrt{d}+2\sqrt{\log N})),\tfrac{1}{2}}$. 

Fix a sample $x \in \set{x^{(1)}, \dots, x^{(N)}}$ and suppose $x \in C_i$ and let $y:=y(x)$ be the corresponding point before the semi-random perturbation, and let $\xbar=x-\mu_i$, $\ybar=y-\mu_i$. For each $i \in [k]$, let $\eii$ be the unit vector along $(\mu_i-\mu'_i)$.  

We first observe that by projecting the Gaussians around $\mu_i, \mu_j$ onto the direction along $\ehat_{ij}=(\mu_i - \mu_j)/\norm{\mu_i - \mu_j}_2$, we have that
\begin{align}
\norm{x-\mu_j}_2^2 - \norm{x-\mu_i}_2^2 &= \iprod{x-\mu_j, \ehat_{ij}}^2 - \iprod{x-\mu_i, \ehat_{ij}}^2 \ge (\abs{\iprod{x-\mu_j, \ehat_{ij}}} - \abs{\iprod{x-\mu_i,\ehat_{ij}}})^2 \nonumber\\
&\ge  (\abs{\iprod{\mu_i-\mu_j, \ehat_{ij}}} - 2\abs{\iprod{x-\mu_i,\ehat_{ij}}})^2 \ge (\Delta \sigma- 2 \abs{\iprod{x-\mu_i, \ehat_{ij}}})^2 \nonumber \\
&\ge (\Delta \sigma- 6 \sigma \sqrt{\log N} )^2 \ge \tfrac{1}{4} \Delta^2\sigma^2 \label{eq:badpoints:1},
\end{align}
where the first inequality follows from \eqref{eq:sr:innerprod:means}, and the second inequality uses $\Delta>12 \sqrt{\log N}$. 

Suppose $x \in C_i$ is misclassified i.e., $\norm{x-\mu'_i}_2 \ge \norm{x-\mu'_j}_2$ for some $j \in [k]\setminus \set{i}$. Then,
\begin{align*}
 \norm{(x-\mu_i)+\mu_i-\mu'_i}_2^2 &\ge  \norm{(x - \mu_j)+(\mu_j - \mu'_j)}_2^2  \\
 2\iprod{x-\mu_i, \mu_i - \mu'_i} - 2\iprod{x-\mu_j, \mu_j - \mu'_j}&\ge \norm{x-\mu_j}_2^2-\norm{x-\mu_i}_2^2+\norm{\mu_j - \mu'_j}_2^2 -\norm{\mu_i - \mu'_i}_2^2 \\
 2\iprod{\xbar, \mu_i - \mu'_i} - 2\iprod{\xbar, \mu_j - \mu'_j} - 2\iprod{\mu_i-\mu_j, \mu_j - \mu'_j} &\ge \tfrac{1}{4}\Delta^2 \sigma^2-\tau^2 \sigma^2 \qquad ~\text{ (from \eqref{eq:badpoints:1})} \\
 2\abs{\iprod{\xbar, \mu_i - \mu'_i}} + 2\abs{\iprod{\xbar, \mu_j - \mu'_j}} &\ge (\tfrac{1}{4}\Delta^2-\tau^2) \sigma^2 - 2\abs{\iprod{\mu_i-\mu_j, \mu_j - \mu'_j}}\\
 \abs[\Big]{\inner[\Big]{\xbar, \frac{\mu_i - \mu'_i}{\norm{\mu_i - \mu'_i}_2}}} + \abs[\Big]{\inner[\Big]{\xbar, \frac{\mu_j - \mu'_j}{\norm{\mu_j - \mu'_j}_2}}} &\ge \frac{\parens[\big]{\frac{\Delta^2}{8} - \frac{\tau^2}{2} - \tau \Delta }\sigma^2}{\tau\sigma} \ge \frac{\Delta^2}{16\tau} \sigma,
\end{align*}
since $\tau < \Delta/(24)$. 
Hence, we have that if $x \in C_i$ is misclassified by $\mu'_1, \dots, \mu'_k$ then 
\begin{equation}\label{eq:temp1}
\abs{\iprod{\xbar,\ehat}} > \sigma \Delta^2/(32\tau) \text{ for some unit vector } \ehat \in \R^d. 
\end{equation}
From \eqref{eq:baddirections} with $\lambda=\Delta^2/(32\tau)$, we get from \eqref{eq:baddirections} that  at most $m$ points in $C_i$ can satisfy \eqref{eq:temp1}. Hence the lemma follows.  
\end{proof}

Next we prove Lemma~\ref{lem:lloyds-iterate}, which quantifies the improvement in every iteration of the Lloyd's algorithm. 
\begin{proof}[Proof of Lemma~\ref{lem:lloyds-iterate}]
Let $C_1, C_2, \ldots ,C_k$ be the partitioning of the indices according to the ground truth clustering of the semi-random instance $\calX$ and $S_1, S_2, \ldots ,S_k$ be the indices of the clustering obtained by using the centers $\mu'_i$. Then $\mu''_i = \frac 1 {|S_i|} \sum_{t \in S_i} x^{(t)}$. Partition $S_i$ into two sets $G_i$ and $B_i$ where $G_i = S_i \cap C_i$ and $B_i = S_i \setminus G_i$. Let $\mu(G_i)$ and $\mu(B_i)$ be the means of the two partitions respectively. Let $\gamma = O(\frac{d\tau^2}{\Delta^4} \max\set{1,\log\parens{\tfrac{3\tau(\sqrt{d}+2\sqrt{\log N})}{\Delta^2}}})$. From Lemma~\ref{lem:badpoints} we know that $|G_i| \geq |C_i| - \gamma$ and $|B_i| \leq  k\gamma$.
Then we have that $\mu''_i = \frac{|G_i|}{|S_i|}\mu(G_i) + \frac{|B_i|}{|S_i|}\mu(B_i)$. Hence, $\|\mu''_i - \mu_i\| \leq \frac{|G_i|}{|S_i|}\|\mu(G_i)-\mu_i\| + \frac{|B_i|}{|S_i|}\|\mu(B_i)-\mu_i\|$. 

We have $\frac{|G_i|}{|C_i|} \geq 1-\frac{\gamma}{|C_i|} \geq 1-\frac{\tau}{64\sqrt{k}\sqrt{d}}$ using the bound on $\Delta$ and $|C_i|= w_i N \geq \frac{d\sqrt{d}}{4 \log(d)}$. Using \eqref{eq:sr:mean:subset} we get that 
$$\frac{|G_i|}{|S_i|}\|\mu(G_i)-\mu_i\| \leq \Big(4 + \frac{2}{\sqrt{1-\frac{\tau}{64\sqrt{k}\sqrt{d}}}} \Big)\sigma \leq \Big(6 + \frac{\tau}{128\sqrt{k}\sqrt{d
}}\Big)\sigma \leq 6\sigma + \frac \tau 8 \sigma.$$

To bound the second term we first show that for each point $x^{(t)} \in B_i$, $\|x^{(t)} - \mu_i\| \leq (\sqrt{d} + 2\sqrt{\log N} + 2\tau) \sigma$. Let $C_j$ be the cluster that point $x^{(t)}$ belongs to. Then 
$$\|x^{(t)}-\mu_i\| \leq \|x^{(t)}-\mu'_i\| + \tau \sigma \leq \|x^{(t)}-\mu'_j\| + \tau \sigma \leq \|x^{(t)}-\mu_j\| + 2\tau \sigma \leq (\sqrt{d} + 2\sqrt{\log N} + 2\tau) \sigma,$$ 
using ~\eqref{eq:sr:length}. Hence, $$\frac{|B_i|}{|S_i|}\|\mu(B_i)-\mu_i\| \leq \frac{|B_i|}{|S_i|} (\sigma \sqrt{d} + \sigma\sqrt{\log N} + 2\tau \sigma) \leq \frac{2k\gamma}{|C_i|}(\sigma \sqrt{d} + \sigma\sqrt{\log N} + 2\tau \sigma) < \frac \tau 8 \sigma.$$
Combining, we get that $\|\mu''_i - \mu_i\| \leq (6 + \frac \tau 4)\sigma \leq \max(6\sigma + \frac \tau 4, \frac \tau 2 \sigma)$.
\end{proof}

\subsection{Initialization}
\label{sec:initialization}
In this section we describe how to obtain the initial centers satisfying the condition in Lemma~\ref{lem:lloyds-iterate}. The final initialization procedure relies on the following subroutine that provides a good initializer if the mean separation is much larger than that in Theorem~\ref{thm:upperbound}. Let $A$ denote the $N \times d$ matrix of data points and $M^*$ be the $N \times d$ matrix where each row of $C$ is equal to one of the means $\mu_i$s of the component to which the corresponding row of $A$ belongs to.

\begin{lemma}
\label{lem:initialization-weak}
In the above notation, for any $\delta>0$ suppose we are given an instance $\calX$ on $N$ points satisfying 
satisfying \eqref{eq:sr:length}-\eqref{eq:baddirections}, with components $C_1, \ldots C_k$ such that $|C_i| \geq \Omega(d+\log(\frac k \delta))$. Let $A$ be the $N \times d$ matrix of data points and $\hat{A}$ be the matrix obtained by projecting points onto the best $k$-dimensional subspace obtained by SVD of $A$. Let $\mu'_i$ be the centers obtained by running an $\alpha$ factor $k$-means approximation algorithm on $\hat{A}$.
Then for every $\mu_i$ there exists $\mu'_i$ such that $\|\mu_i - \mu'_i\| \leq 20\sqrt{k\alpha} \frac{\|A-M^*\|}{\sqrt{N\wmin}}$.
\end{lemma}
\begin{proof}
Let $\hat{A}$ denote the matrix obtained by projecting $A$ onto the span of its top $k$ right singular vectors. Furthermore, let $\nu_1, \ldots \nu_k$ be the centers obtained by running a $9$-approximation algorithm for $k$-means on the instance $\hat{A}$. We know that the optimal $k$-means solution for $\hat{A}$ is at most $\|\hat{A}-M^*\|^2_F$. Since both $\hat{A}$ and $M^*$ are rank $k$ matrices, we get that $\|\hat{A}-M^*\|^2_F \leq 2k\|\hat{A}-M^*\|^2_2 \leq 2k(\|\hat{A}-A\|^2_2 + \|A-M^*\|^2_2)$. Since $\hat{A}$ is the best rank $k$ approximation to $A$ we also have that $\|\hat{A}-A\|^2_2 \leq \|A-M^*\|^2_2$. Hence, $\|\hat{A}-M^*\|^2_F \leq 4k\|A-M^*\|^2_2$. Hence, the cost of the solution using centers $\nu_i$s must be at most $36 k\sigma^2N$~(using \ref{eq:sr:variance:all}).

Next, suppose that there exists $\mu_i$ such that for all $j$, $\|\mu_i - \nu_j\| > 20\sqrt{k\alpha} \frac{\|A-M^*\|}{\sqrt{N\wmin}}$. let's compute the cost paid by the points in component $C_i$ in the clustering obtained via the approximation algorithm. For any $x \in C_i$ let $\nu_x$ be the center that it is closest to. Then the cost is at least $\sum_{x \in C_i} \|x-\nu_x\|^2 \geq \sum_{x \in C_i} \frac 1 2 \|\mu_i - \nu_x\|^2 - \|x-\mu_i\|^2$. The first summation is at least $\frac 1 2 |N\wmin| (400 \alpha k \frac{\|A-M^*\|^2}{N\wmin}) > 200k\alpha \|A-M^*\|^2$.
The second summation is at most $\sum_{x \in C_i} \|x-\mu_i\|^2 \leq \sum_{i} \sum_{x \in C_i} \|x-\mu_i\|^2 = \|\hat{A}-M^*\|^2_F \leq 4k\|A-M^*\|^2$. Hence, we reach a contradiction to the fact that the solution obtained via $\nu_i$s is an $\alpha$-approximation to the optimal cost.
\end{proof}

The proof of the above theorem already provides a good initializer provided $\Delta$ is larger than $\sqrt{k \frac{\log N}{\wmin}}$ and one uses a constant factor approximation algorithm for $k$-means~\cite{AhmadianNSW16}. Furthermore, if $\Delta$ is larger than $\sqrt{k \log k \frac{\log N}{\wmin}}$, then one can instead use the simpler and faster $k$-means++ approximation algorithm~\cite{AV07}. The above lemma has a bad dependence on $\wmin$. 
However, using the Boosting technique of~\cite{KK10} we can reduce the dependence to $\Delta > 25 \sqrt{k \log N}$ and hence prove Proposition~\ref{prop:initialization-strong}. We provide a proof of this in the Appendix.

\section{Lower Bounds for Semi-random GMMs}
\label{sec:lower-bound}
We prove the following theorem. 
\begin{theorem}\label{thm:lowerbound}
For any $d,k \in Z_+$, there exists $N_0=\poly(d,k)$ and a universal constant $c_1>0$ such that the following holds for all $N\ge N_0$ and $\Delta$ such that  $\sqrt{\log N} \le \Delta \le d/(64 \log d)$.  There exists an instance $\calX$ on $N$ points in $d$ dimensions with planted clustering $C_1, \dots, C_k$ generated by applying semi-random perturbations to points generated from a mixture of spherical Gaussians with means $\mu_1, \mu_2, \dots, \mu_k$, covariance $\sigma^2 I$ and weights being $1/k$ each, with separation $\forall i \ne j \in [k], \norm{\mu_i - \mu_j}_2 \ge \Delta \sigma$, such that any locally optimal $k$-means clustering solution $C'_1, C'_2, \dots, C'_k$ of $\calX$ satisfies w.h.p. 
$$ \min_{\pi \in \text{Perm}_k}\sum_{i=1}^k |C'_{\pi(i)} \triangle C_i| \ge \frac{c_1 k d}{\Delta^4}.$$
It suffices to set $N_0(d,k):=c_0 k^2 d^{3/2}\log^2(kd)$, where $c_0>0$ is a sufficiently large universal constant.  
\end{theorem}
\begin{remark}
Note that the lower bound also applies in particular to the more general semi-random model in Definition~\ref{def:srmodel}; in this instance, the points are drawn i.i.d. from the mixture of spherical Gaussians, before applying semi-random perturbations. Further, this lower bound holds for any {\em locally optimal solution}, and not just the optimal solution.   
\end{remark}

The lower bound construction will pick an arbitrary $\Omega(d/\Delta^4)$ points from $k/2$ clusters, and carefully choose a semi-random perturbation to all the points so that these $\Omega(kd/\Delta^4)$ points are misclassified.   
We start with a simple lemma that shows that an appropriate semi-random perturbation can move the mean of a cluster by an amount $O(\sigma)$ along any fixed direction.
\begin{lemma}\label{lem:movingcenter}
Consider a spherical Gaussian in $d$ dimensions with mean $\mu$ and covariance $\sigma^2 I$, and let $\ehat$ be a fixed unit vector. Consider the semi-random perturbation given by 
$$\forall y \in \R^d, h(y)= \begin{cases} \mu & ~\text{ if } \iprod{y-\mu, \ehat}<0\\ y & ~\text{ otherwise }  \end{cases}.$$
Then we have $\E[h(y)]=\mu+\tfrac{1}{\sqrt{2\pi}} \sigma \ehat$. 
\end{lemma}
\begin{proof}
We assume without loss of generality that  $\mu=0, \sigma=1$ (by shifting and scaling) and $\ehat=(1,0,0,\dots, 0) \in \R^d$ (by the rotational symmetry of a spherical Gaussian). Let $\gamma$ be the p.d.f. of the standard Gaussian in $d$ dimensions with mean $0$, and $\gamma'(y)$ be the distribution on $y$ conditioned on the event $[y(1)=\iprod{y, \ehat}>0]$. First, we observe that $\E[h(y)| y_1 <0 ]=0$ from construction, and $\E[h(y)|y_1>0]=\E_{y \sim \gamma'(y)}[y]$. Further, since the $(d-1)$ co-ordinates of $y$ orthogonal to $\ehat$ are independent of $y_1$,
\begin{align*}
\E[h(y)]&=\Pr[y_1<0] \E[h(y)|y_1<0] +\Pr[y_1>0] \E[h(y)|y_1>0] = \frac{1}{2} \E[y_1|y_1>0]  \ehat\\
\E[h(y)]-\mu & \E[h(y)]=\parens[\Big]{\frac{1}{2\sqrt{2\pi}}\int_{-\infty}^{\infty} |y_1| \exp(-y_1^2/2) ~dy_1} \ehat = \frac{\sigma}{\sqrt{2\pi}} \ehat.
\end{align*}
\end{proof}

\paragraph{Construction.}
Set $m:=c_1 d/\Delta^4$ for some appropriately small constant $c_1 \in (0,1)$. We assume without loss of generality that $k$ is even (the following construction also works for odd $k$ by leaving the last cluster unchanged). 
We pair up the clusters into $k/2$ pairs $\set{(C_1,C_2), (C_3, C_4), \dots, (C_{k-1},C_k)}$, and we will ensure that $m$ points are misclassified in each of the $k/2$ clusters $C_1, C_3, \dots, C_{k-1}$. The parameters of the mixture of spherical Gaussians $\calG$ are set up as follows. For each $i \in {1,3,5,\dots,k-1}$, $\norm{\mu_i-\mu_{i+1}}_2 = \Delta \sigma$, and all the other inter-mean distances (across different pairs) are at least $M\sigma$ which is arbitrarily large (think of $M \mapsto \infty$).	 
 
\begin{itemize}
\item Let for any $i \in \set{1,3,\dots, k-1}$, $Z_{i} \subset C_{i}$ be the first $m$ points in cluster $C_{i}$ respectively among the samples $y^{(1)},\dots, y^{(N)}$ drawn from $\calG$ (these $m$ points inside the clusters can be chosen arbitrarily). Set $Z_i =\emptyset$ for $i \in \set{2,4,\dots,k}$.   
\item For each $i \in \set{1,3,\dots,k-1}$, set $\ehat_i$ to be the unit vector along 
$u_i=\frac{1}{\sigma \sqrt{md}}\sum_{y \in Z_{i}} (y-\mu_{i})$.
\item For each $i \in \set{1,3,\dots, k-1}$ apply the following semi-random perturbation given by Lemma~\ref{lem:movingcenter} to points in cluster $C_{i+1}$ along $\ehat_i$, i.e., each point $y^{(t)} \in C_{i+1}$
$$x^{(t)} = h(y^{(t)})=\begin{cases} \mu_{i+1} & \text{if } \iprod{y^{(t)}-\mu_{i+1},\ehat_i}<0\\ y^{(t)} & \text{otherwise} \end{cases}.$$
\end{itemize}

Note that the semi-random perturbations are only made to points in the even clusters (based on a few points in its respective odd cluster). The lower bound proof proceeds in two parts. Lemma~\ref{lem:meanmovement} (using Lemma~\ref{lem:movingcenter}) and Lemma~\ref{lem:smallimpurity} shows that in any $k$-means optimal clustering the means of each even cluster $C_i$ moves by roughly $\Omega(\sigma) \cdot \ehat_{i-1}$. Lemma~\ref{lem:misclassification} then shows that these means will classify all the $m$ points in $Z_{i-1}$ {\em incorrectly} w.h.p. In this proof w.h.p. will refer to a probability of at least $1-o(1)$ unless specified otherwise (this can be made $1-1/\poly(m,k)$ by choosing suitable constants).

We start with two simple concentration statements about the points in $Z_i$ (from Lemma~\ref{lem:lengthconc} and Lemma~\ref{lem:sr:innerprod:means}). We have with probability at least $1-1/(mk)$, 
\begin{align}
\forall i \in \set{1,3,k-1}, ~\forall t \in Z_{i}~~ \norm{x^{(t)}-\mu_i}_2 \le \sigma(\sqrt{d}+ 2\sqrt{\log (mk)}) \label{eq:lowerbound:lengthconc}\\
\forall i \in \set{1,3,k-1}, ~\forall t \in Z_{i}~~ \abs{\iprod{x^{(t)}-\mu_i, \mu_i - \mu_{i+1}}} \le 2 \sqrt{\log (mk)} \Delta \sigma^2 \label{eq:lowerbound:innerprod}
\end{align}

We start with a couple of simple claims about the unit vectors $\ehat_1, \ehat_3,\dots, \ehat_{k-1}$. 
\begin{lemma}\label{lem:lengthofu}
In the above construction, for every $i \in \set{1,3,\dots, k-1}$ we have w.h.p. $\norm{\ehat_i-u_i}^2_2 \le 6 \sqrt{m\log (mk)/d}$. Further, 
for each $x \in Z_i$, we have $\iprod{x-\mu_i, \ehat_i} \ge \tfrac{1}{2} \sigma \sqrt{d/m}$.
\end{lemma}
\begin{proof}
Let us fix an $i \in \set{1,3,\dots, k-1}$. Let $y^{(1)}, y^{(2)}, \dots, y^{(m)} \in Z_i$ and $\ybar^{(t)}=y^{(t)}-\mu_i$. 

From \eqref{eq:lowerbound:lengthconc}, we know that w.h.p., $\norm{\ybar^{(t)}}_2 \le \sigma(\sqrt{d}+2\sqrt{\log m})~ \forall t \in [m]$. 
Fix $t \in [m]$, and let $Q(t)= \sum_{t' \in [m]\setminus \set{t} } \iprod{\ybar^{(t)}, \ybar^{(t')}}$. For $t' \ne t$, due to independence and spherical symmetry,  $\tfrac{1}{\norm{\ybar^{(t)}}_2}\iprod{\ybar^{(t)}, \ybar^{(t')}}$ is distributed as a normal r.v. with mean $0$ and variance $\sigma^2$. Further, $Q(t)/\norm{y^{(t)}}_2$ is distributed as a normal r.v. with mean $0$ and variance $\sigma^2 m$. Hence, 
\begin{equation}\label{eq:qt}
Q(t) = \norm{y^{(t)}}_2 \cdot \sum_{t' \in [m]\setminus \set{t}} \inner[\big]{\ybar^{(t')}, \frac{\ybar^{(t)}}{\norm{\ybar^{(t)}}_2}} \le \sigma^2(\sqrt{d}+\sqrt{\log (mk)})\cdot 2\sqrt{m \log (mk)},  
\end{equation}
with probability at least $1-1/(mk)^2$. Hence, w.h.p. $Q(t) \le 4\sigma^2 \sqrt{d m \log (mk)}$ for all $t \in [m]$. 

For the first part, we see that
$$\norm{u_i}_2^2=\frac{1}{\sigma^2 md} \parens[\Big]{\sum_{t \in [m]} \norm{\ybar^{(t)}}_2^2  + 2\sum_{t \ne t' \in [m]}\iprod{\ybar^{(t)},\ybar^{(t')}} }=\frac{1}{\sigma^2 md} \parens[\Big]{\sum_{t \in [m]} \norm{\ybar^{(t)}}_2^2  + 2\sum_{t \in [m]} Q(t)}.$$
 Along with \eqref{eq:lowerbound:lengthconc}, the bound on $Q(t)$ and $\E[\norm{y^{(t)}}_2^2]=d \sigma^2$, this implies 
\begin{align*}
\abs{\norm{u_i}_2^2 - 1}&\le \frac{1}{md} \parens{4m \sqrt{d \log (mk)}+4m \log(mk)+ 4m \sqrt{dm \log (mk)}} \text{ w.p. at least } 1- 1/(mk) \\
\abs{\norm{u_i}_2^2 - 1}&\le 6 \sqrt{\frac{m \log (mk)}{d}}  ~\text{ with probability at least } 1-1/(mk). 
\end{align*}
Since $\ehat_i$ is the unit vector along $u_i$, and performing a union bound over all $i$ we have that w.h.p.,  $\norm{\ehat_i - u_i}^2_2 \le 6  \sqrt{m \log (mk)/d}$. 

For the furthermore part, suppose $x=y^{(t)}$ for some $t \in [m]$ then 
\begin{align*}
\iprod{\xbar, \ehat_i}&= \frac{1}{\norm{u_i}_2 \sqrt{md}} \sum_{t' \in [m]} \iprod{y^{(t)}, y^{(t')}} \ge \frac{1}{\norm{u_i}_2 \sqrt{md}} \parens[\Big]{ \norm{\ybar^{(t)}}_2^2 -  \sum_{t' \ne t} \abs{\iprod{\ybar^{(t)}, \ybar^{(t')}}}}\\
&\ge \frac{\sigma^2}{\norm{u_i}_2 \sqrt{md}}\parens[\Big]{(d-\sqrt{d \log (mk)}) - Q(t) }\ge \frac{\sigma^2}{\norm{u_i}_2 \sqrt{md}} \parens[\big]{d - 4\sqrt{dm \log (mk)}}\\
&\ge \frac{\sigma^2}{4} \sqrt{d}{m}, 
\end{align*}
since $64 m \log m \le d$ and $\norm{u_i}_2 \le 2$ w.h.p. 
\end{proof}

Let $\tmu_1, \dots, \tmu_k$ be the (empirical) means of the clusters in the planted clustering $C_1, C_2, \dots, C_k$ after the semi-random perturbations. The following lemma shows that $\norm{\tmu_i - \mu_i}_2 \le \sigma$. 

\begin{lemma}\label{lem:meanmovement}
There exists a universal constant $c_3>0$ s.t. for the semi-random instance $\calX$ described above, we have that w.h.p. 
\begin{align*}
\forall i \in [k], ~ \tmu_i &=\begin{cases} \mu_i + \tfrac{1}{\sqrt{2\pi}}\sigma \ehat_{i-1} + z_{i} & \text{ if } i \text{ is even} \\
\mu_i + z_{i} & \text{ if } i \text{ is odd}
\end{cases}, ~~\text{where }~  \norm{z_i}_2 \le c_3 \sigma \sqrt{\frac{dk}{N}}.
\end{align*}
\end{lemma}
\begin{proof}
The lemma follows in a straightforward way from Lemma~\ref{lem:movingcenter} and by standard concentration bounds. Firstly, the clusters $C_i$ for odd $i$ are unaffected by the perturbation. Hence, $\E_{x \in C_i}[x]=\mu_i$ and from Lemma~\ref{lem:Gaussian-mean-concentration}, the empirical mean of the points in $C_i$ (there are at least $N/(2k)$ of them w.h.p.) gives the above lemma.  Consider any even $i$. From Lemma~\ref{lem:movingcenter}, the semi-random perturbation applied to the points in $C_{i}$ along the direction $\ehat_{i-1}$ ensures that $\E_{x \in C_i}[x]=\mu_i +\tfrac{\sigma}{\sqrt{2\pi}} \ehat_{i-1}$. Again by Lemma~\ref{lem:Gaussian-mean-concentration} applied to the points from $C_i$, the lemma follows.  
\end{proof}

The following lemma shows that if $C_i, C'_i$ are close, then the empirical means are also close.  
\begin{lemma}\label{lem:smallimpurity}
Consider any cluster $C_i$ of the instance $\calX$, and let $C'_i$ satisfy $|C'_i \triangle C_i| \le m'$. Suppose $\tmu_i$ and $\mu'_i$ are the means of clusters $C_i$ and $C'_i$ respectively, then 
$$\norm{\mu'_i-\tmu_i}_2 \le 4 \sigma \cdot  \frac{m'}{|C_i|}\parens{\sqrt{d}+2\sqrt{\log N}+\Delta} .$$    
\end{lemma}
\begin{proof}
Let $i$ be even (an even cluster).   
First, we note that from our construction, all the points in $C'_i \setminus C_i \in C_{i-1}$ w.h.p., since the distance between the means $\norm{\mu_i- \mu_j}_2 \geq M \sigma$ when $j \notin \set{i-1, i}$, for $M$ that is chosen to be appropriately large enough. Further, $\norm{\mu_i - \mu_{i-1}}_2 =\Delta \sigma$. Let $\xbar=x-\mu_i$ if $i \in C_i$ and $\xbar=x-\mu_j$ if $x \in C_j$. Hence w.h.p.,
$$\forall x \in C'_i \cup C_i, ~~ \norm{x - \mu_i}_2 \le \Delta \sigma+ \norm{\xbar}_2 \le  \Delta \sigma+ (\sqrt{d}+2\sqrt{\log N}) \sigma.$$
Further, $\tmu_i$ is the empirical mean of all the points in $C_i$. Let $\delta=m'/|C_i|$.  
\begin{align*}
\mu'_i -\mu_i&=\frac{\sum_{x \in C_i} (x-\mu_i)}{|C'_i|} - \frac{\sum_{x \in C_i \setminus C'_i} (x-\mu_i)}{|C'_i|} + \frac{\sum_{x \in C'_i \setminus C_i} (x-\mu_i)}{|C'_i|}  \\
\mu'_i -\tmu_i&= (\mu'_i - \mu_i)- (\tmu_i - \mu_i) = (\mu'_i-\mu_i) + \frac{\sum_{x \in C_i} (x-\mu_i)}{|C_i|} \\
\text{Hence, } \mu'_i - \tmu_i &=\parens[\big]{\tfrac{|C_i|}{|C'_i|}-1} (\tmu_i-\mu_i) - \frac{1}{|C'_i|}\sum_{x \in C_i \setminus C'_i} (x-\mu_i) + \frac{1}{|C'_i|}\sum_{x \in C'_i \setminus C_i} (x-\mu_i)\\
\norm{\mu'_i - \tmu_i}_2 &\le \parens[\big]{\frac{\delta}{1-\delta}} \norm{\tmu_i-\mu_i}_2 + \parens[\big]{\frac{2\delta}{1-\delta}} \max_{x \in C_i \cup C'_i}\norm{x -\mu_i}\\
& \le \parens[\big]{\frac{2\delta \sigma}{1-\delta}} (1+ \Delta+\sqrt{d}+2\sqrt{\log N}) \le 4 \delta \sigma (\Delta+\sqrt{d}+2\sqrt{\log N}),
\end{align*}
where $\norm{\mu_i-\tmu_i}_2$ is bounded because of Lemma~\ref{lem:meanmovement}. 
A similar argument follows when $i$ is odd. 
\end{proof}

The following lemma shows that the Voronoi partition about $\tmu_1, \dots, \tmu_k$ (or points close to it) incorrectly classify all points in $Z_i$ for each $i \in [k]$.   
\begin{lemma}\label{lem:misclassification}
Let $\mu'_1,\mu'_2, \dots, \mu'_k$ satisfy $\norm{\mu'_i - \tmu_i}_2 \le \sigma/(16 \sqrt{m}(1+2\sqrt{\tfrac{\log N}{d}}))$, where $\tmu_i$ is the empirical mean of the points in $C_i$.  Then, we have w.h.p. that for each $i \in \set{1,3,\dots, k-1}$, $\norm{x-\mu'_i}_2^2 > \norm{x-\mu'_{i+1}}_2^2$, i.e., every point $x \in Z_i$ is misclassified. 
\end{lemma}
\begin{proof}
Let $i$ be odd, and consider a point $x$ in $Z_i$, and let $\xbar=x-\mu_i$.
\begin{align*}
&\norm{x-\mu'_i}_2^2 - \norm{x-\mu'_{i+1}}_2^2 = \norm{(x-\mu_i)+\mu_i-\mu'_i}_2^2 -  \norm{(x - \mu_{i+1})+(\mu_{i+1} - \mu'_{i+1})}_2^2  \\
&\quad =\norm{x-\mu_i}_2^2 - \norm{x-\mu_i + (\mu_i-\mu_{i+1})}_2^2
+ 2\iprod{x-\mu_{i}, \mu_i - \mu'_i} - 2\iprod{x-\mu_{i+1}, \mu_{i+1} - \mu'_{i+1}}\\
&\quad\quad\quad +\norm{\mu_{i} - \mu'_i}_2^2-\norm{\mu_{i+1} - \mu'_{i+1}}_2^2 \\
& \quad \ge 2\iprod{\xbar, \mu_i - \mu'_i} + 2\iprod{\xbar, \mu'_{i+1} - \mu_{i+1}}+2\iprod{\xbar,\mu_{i+1} - \mu_{i}}- \Delta^2 \sigma^2 - 2\iprod{\mu_i-\mu_{i+1}, \mu_{i+1} - \mu'_{i+1}}- \sigma^2 \\
& \quad \ge 2\iprod{\xbar, \mu_i - \mu'_i} + 2\iprod{\xbar, \mu'_{i+1} - \mu_{i+1}}- 4 \Delta \sqrt{\log (mk)} \sigma^2- \Delta^2 \sigma^2 - 2\Delta \sigma^2- \sigma^2,
\end{align*}
where the last inequality follows from \eqref{eq:lowerbound:innerprod}.
From Lemma~\ref{lem:meanmovement}, we have 
\begin{align*}
\mu'_{i+1}-\mu_{i+1}&=(\tmu_{i+1} - \mu_{i+1})+(\mu'_{i+1}-\tmu_{i+1})=\tfrac{1}{\sqrt{2\pi}}\sigma \ehat_i + z'_{i+1},\\
&\text{where }~ \norm{z'_{i+1}}_2 \le \norm{z_{i+1}}_2+\norm{\mu'_{i+1}-\tmu_{i+1}}_2 \le \sigma \cdot \frac{1}{(12 \sqrt{m}(1+2\sqrt{\log N/d})}, 
\end{align*}
since $N/\sqrt{\log N} \ge C d^{3/2} km$ for some appropriately large constant $C>0$. Similarly $\mu'_i - \mu_i=z'_i$, where $\norm{z'_i}_2 \le \sigma/(12 \sqrt{m}(1+\sqrt{\log N/d}))$. 
Hence, simplifying and applying Lemma~\ref{lem:lengthofu} we get
\begin{align*}
\norm{x-\mu'_i}_2^2 - \norm{x-\mu'_{i+1}}_2^2 &\ge \tfrac{2}{\sqrt{2\pi}}\iprod{\xbar,\ehat} \sigma - 2\abs{\iprod{\xbar,z'_i}}-2\abs{\iprod{\xbar,z'_{i+1}}}- 4 \Delta \sqrt{\log (mk)} \sigma^2-\Delta^2 \sigma^2- \sigma^2\\
&\ge \sqrt{\frac{d}{2\pi m}} \cdot \sigma^2 - 2\norm{x}_2 (\norm{z'_i}_2+\norm{z'_{i+1}}_2)-4 \Delta^2 \sigma^2 \\
&\ge \sigma^2 \sqrt{\frac{d}{2\pi m}} - \sigma^2 \sqrt{\frac{d}{9m}} - 4 \sigma^2 \Delta^2 >0,
\end{align*}
since $m \le c d/\Delta^4$ for some appropriate constant $c$ (say $c=16 \pi$).
\end{proof}

\begin{proof}[Proof of Theorem~\ref{thm:lowerbound}]
Let $C'_1, \dots, C'_k$ be a locally optimal $k$-means clustering of $\calX$, and suppose $\sum_i |C'_i \triangle C_i| < mk/2$ (for sake of contradiction). For each $i \in [k]$, let $\tmu_i$ be the empirical mean of $C_i$ and $\mu'_i$ be the empirical mean of $C'_i$. Since $C'_1, \dots, C'_k$ is a locally optimal clustering, the Voronoi partition given by $\mu'_1, \dots,\mu'_k$ classifies all the points in agreement with $C'_1, \dots, C'_k$. 

We will now contradict the local optimality of the clustering $C'_1, \dots, C'_k$.
Every cluster $C_i$ has at least $N/(2k)$ points w.h.p. Hence, for each $i \in [k]$, from Lemma~\ref{lem:smallimpurity} we have 
\begin{align*}
\norm{\mu'_i - \tmu_i}_2 &\le \sigma (\sqrt{d}+2\sqrt{\log N}+\Delta)\cdot \frac{4|C_i \triangle C'_i|}{\tfrac{N}{2k}} \le \sigma \cdot \frac{8k^2 m (\sqrt{d}+2\sqrt{\log N}+\Delta)}{N} \\
&\le \frac{\sigma}{16 \sqrt{m}(1+\sqrt{(\log N)/d})}.
\end{align*}
However, from Lemma~\ref{lem:misclassification}, every point in $\cup_{i \in [k]} Z_i$ is misclassified by $\mu'_1, \mu'_2, \dots, \mu'_k$, i.e., the clustering given the Voronoi partition around $\mu'_1, \dots, \mu'_k$ differs from $C_1, \dots, C_k$ on at least $mk/2$ points in total. But $\sum_{i \in [k]} |C'_i \triangle C_i|<mk/2$. Hence, this contradicts the local optimality of the clustering $C'_1, \dots, C'_k$. 
\end{proof}

\section{Conclusion}
In this work we initiated the study of clustering data from a semi-random mixture of Gaussians. We proved that the popular Lloyd's algorithm achieves near optimal error.   
The robustness of the Lloyd's algorithm for the semi-random model suggests a theoretical justification for its widely documented success in practice. Similar robust analysis under stronger adversaries or for related heuristics such as the EM algorithm will significantly improve the gap between our theoretical understanding and observed practical performance of these algorithms. 
It would also be interesting to extend our results to study semi-random variants for other statistical models that are popular in machine learning.

\bibliographystyle{alpha}
\bibliography{aravind}

\appendix

\section{Standard Properties of Gaussians}

\newcommand{\phit}{\Phi}
\newcommand{\iphit}{\Phi^{-1}}
\begin{lemma}\label{lem:gaussian1d}
Suppose $x \in \R$ be generated according to $N(0,\sigma^2)$, let $\phit(t)$ represent the probability that $x >t$, and let $\iphit{y}$ represent the quantile $t$ at which $\phit(t)\le y$.
Then 
\begin{equation}\label{eq:tails}
\frac{\frac{t}{\sigma}}{(\frac{t^2}{\sigma^2}+1)} e^{-\frac{t^2}{2 \sigma^2}}  \le \phit(t) \le \frac{\sigma}{t} e^{-\frac{t^2}{2 \sigma^2}}.
\end{equation}
Further, there exists a universal constant $c \in (1,4)$ such that
\begin{equation}\label{eq:tails2}
 \frac{1}{c} \sqrt{\log (1/y)} \le \frac{t}{\sigma}  \le c \sqrt{\log (1/y)}.
\end{equation}
\end{lemma}

Let $\gamma_d$ be the Gaussian measure associated with a standard Gaussian with mean $0$ and variance $1$ in each direction. 
We start with a simple fact about the probability mass of high-dimensional spherical Gaussians being concentrated at around $\sqrt{d} \sigma$.

Using concentration bounds for the $\chi^2$ random variables, we have the following bounds for the lengths of vectors picked according to a standard Gaussian in $d$ dimensions (see (4.3) in \cite{laurent2000}).
\begin{lemma}\label{lem:lengthconc}
For a standard Gaussian in $d$ dimensions (mean $0$ and variance $\sigma^2$ in each direction), and any $t>0$  
\begin{align*}
\Pr_{x \sim \gamma_d}\Big[ \norm{x}^2 \ge \sigma^2(d + 2 \sqrt{d t}+2 t) \Big] &\le e^{-t}.\\
\Pr_{x \sim \gamma_d}\Big[ \norm{x}^2 \le \sigma^2(d - 2 \sqrt{d t}) \Big] &\le e^{-t}.
\end{align*}
\end{lemma}

The following lemma follows from Lemma~\ref{lem:lengthconc} and a simple coupling to a spherical Gaussian with variance $\sigma^2 I$.
\begin{lemma}\label{lem:amphelp:length}
Consider any points $y^{(1)}, \dots, y^{(N)}$ drawn from a Gaussian with mean $0$ and variance at most $\sigma^2$ in each direction. Then with high probability we have
$$\forall \ell \in [N], \norm{y^{(\ell)}}_2 \le \sigma(\sqrt{d}+2\sqrt{\log N})  .$$ 
\end{lemma}
\begin{proof}
Consider a random vector $z \in \R^d$ generated from a Gaussian with mean $0$ and variance $\sigma^2$ in each direction. From Lemma~\ref{lem:lengthconc},
$$Pr[\norm{z}_2 \ge \sigma(\sqrt{d}+2\sqrt{\log N})]=Pr[\norm{z}^2_2 \ge \sigma^2(d+4\sqrt{d\log N}+4\log N)] \le \exp(- 2 \log N) < N^{-2} .$$
Fix $\ell \in [N]$. By a simple coupling to the spherical Gaussian random variable $z$ we have
$$Pr[\norm{y^{(\ell)}}_2 \ge \sigma(\sqrt{d}+2\sqrt{\log N})]\le Pr[\norm{z} \ge \sigma(\sqrt{d}+2\sqrt{\log N})] < N^{-2} .$$
By a union bound over all $\ell \in [N]$, the lemma follows. 
\end{proof}

\begin{lemma}[\cite{vershynin2010introduction}, Proposition 5.10]
\label{lem:Gaussian-mean-concentration}
Let $Y_i \sim N(\mu,\sigma^2I_{d \times d})$ for $i=1,2,\ldots N$ where $N = \Omega(\frac{d + \log(\frac 1 \delta)}{\epsilon^2})$. Then, with probability at least $1-\delta$ we have that 
$$
\norm{\frac 1 N \sum_{i=1}^N Y_i - \mu}_2 \leq \sigma \epsilon.
$$
\end{lemma}

\begin{lemma}[\cite{vershynin2010introduction}, Corollary 5.50]
\label{lem:Gaussian-variance-concentration}
Let $Y_i \sim N(\mu,\sigma^2I_{d \times d})$ for $i=1,2,\ldots N$ where $N = \Omega(\frac{d + \log(\frac 1 \delta)}{\epsilon^2})$. Then, with probability at least $1-\delta$ we have that 
$$
\norm{\frac 1 N \sum_{i=1}^N (Y_i-\mu)(Y_i-\mu)^T - \sigma^2I} \leq \sigma\epsilon.
$$
\end{lemma}

\section{Proof of Proposition~\ref{prop:initialization-strong}}
The proof will follow the outline in~\cite{KK10}. Given $N$ points from a semi random mixture $\calX$, we first randomly partition them into two sets $S_1$ and $S_2$ of equal size. Let $T_1, \ldots ,T_k$ be the partition induced by the true clustering over $S_1$ and $T'_1, \ldots ,T'_k$ be the partition induced over $S_2$. Furthermore, let $A$ be the $\frac N 2 \times d$ matrix consisting of points in $S_1$ as rows and $C$ be the $\frac N 2 \times d$ matrix of the corresponding true centers. It is easy to see that with probability at least $1-\delta$, we will have that 
\begin{equation}\label{eq:boost:partition:size}
\forall r \in [k], \min(|T_r|, |T'_r|) \geq \frac{|C_r|}{4}.
\end{equation}

Assuming that equations~(\ref{eq:sr:length}) to (\ref{eq:baddirections}) hold with high probability, we next prove that the following conditions will also hold with high probability
\begin{align}
\max_{v: \|v\|=1} \frac 1 {|T_r|} \sum_{x \in T_r} [(x-\mu_r) \cdot v]^2 & \leq  4\sigma^2, \forall r \in [k]\label{eq:boost:var1}\\
\max_{v: \|v\|=1} \frac 1 {|T'_r|} \sum_{x \in T'_r} [(x-\mu_r) \cdot v]^2 & \leq  4\sigma^2, \forall r \in [k] \label{eq:boost:var2}\\
\|\frac 1 {|T_r|} \sum_{x \in T_r} x - \mu_i\| & \leq  8\sigma, \forall r \in [k] \label{eq:boost:mean}\\
\|A-C\|^2 & \leq  4\sigma^2 N \label{eq:boost:varall}
\end{align}

To prove (\ref{eq:boost:var1}) notice that $\frac 1 {|T_r|} \sum_{x \in T_r} [(x-\mu_r) \cdot v]^2 \leq \frac  4 {|C_r|} \sum_{x \in C_r} [(x-\mu_r) \cdot v]^2 \leq 16\sigma^2$~(using \ref{eq:sr:variance}). Similarly, (\ref{eq:boost:var2}) follows. The proof of (\ref{eq:boost:mean}) follows directly from (\ref{eq:sr:mean:subset}). Finally notice that $\|A-C\|^2 = \max_{v:\|v\|=1} \sum_r \sum_{x \in T_r} [(x-\mu_r)\cdot v]^2 \leq \max_{v:\|v\|=1} \sum_r \sum_{x \in C_r} [(x-\mu_r)\cdot v]^2 \leq 4\sigma^2 N$~(using (\ref{eq:sr:variance:all})).
 
In the analysis below we will assume that the above equations are satisfied by the random partition. Define a graph $G = (A \cup B, E)$ where the edge set consists of any pair of points that have a distance of at most $\gamma = 4\sigma(\sqrt{d} + \sqrt{\log N})$. Notice that from the definition of $\gamma$, any two points from the same true cluster $C_r$ will be connected by an edge in $G$~(using \ref{eq:sr:length}). Next we map the points in $A$ to a new $\frac N 2$ dimensional space as follows. For any row $A_i$ of $A$ define $A'_{i,j} = (A_i-\mu)\cdot(B_j-\mu)$ if $A_i$ and $B_j$ are in the same connected component of $G$. Otherwise, define $A'_{i,j} = L$ where $L$ is a large quantity. Here $\mu$ denotes the mean of the points in the component in $G$ to which $A_i$ belongs to. Let $\theta_r$ denote the mean of the points in $T_r$ in the new space. We will show that the new mapping amplifies the mean separation.  
\begin{lemma}
\label{lem:boost-mean-separation}
For all $r \neq s$, $\|\theta_r - \theta_s\| \geq \Omega(\sqrt{|N \wmin|} k \log N) \sigma^2$.
\end{lemma}
\begin{proof}
We can assume that points in $T_r, T'_r$ and $T_s, T'_s$ belong to the same connected component in $G$. Otherwise, $\|\theta_r - \theta_S\| > L$. Let $Q$ be the component to which $T_r$ and $T_s$ belong with $\mu$ being the mean of the points in $Q$. Then, $\|\theta_r - \theta_s\|^2 \geq \sum_{B_j \in Q} [(\mu_r - \mu_s)\cdot(B_j-\mu)]^2$. Notice that $(\mu_r - \mu_s)\cdot(\mu_r - \mu_s) = (\mu_r - \mu)\cdot(\mu_r - \mu_s) - (\mu_s - \mu)\cdot(\mu_r - \mu_s)$. Hence, one of the two terms is at least $\frac 1 2 \|\mu_r - \mu_s\|^2$ in magnitude. Without loss of generality assume that $|(\mu_r - \mu)\cdot(\mu_r - \mu_s)| \geq \frac 1 2 \|\mu_r - \mu_s\|^2 \geq \frac{125^2}{2} k \log N$. 

Now, $\|\theta_r - \theta_s\|^2 \geq \sum_{B_j \in T'_r} [(\mu_r - \mu_s)\cdot(B_j-\mu)]^2 = \sum_{B_j \in T'_r} [(\mu_r - \mu_s)\cdot(\mu_r - \mu) - (\mu_r - \mu_s)\cdot(\mu_r - B_j)]^2 \geq \frac 1 2 |B_j| [(\mu_r-\mu_s)\cdot(\mu_r-\mu)]^2 - \sum_{B_j \in T'_r} [(\mu_r - \mu_s)\cdot(\mu_r - B_j)]^2$. The first term is at least $\frac{|T'_r|}{8} \|\mu_r - \mu_s\|^4$ and the second term (in magnitude) is at most $4|T'_r| \|\mu_r-\mu_s\|^2 \sigma^2$~(using \ref{eq:boost:var2}). Substituting the bound on $\|\mu_r-\mu_s\|$ and using \ref{eq:boost:partition:size}, we get that $\|\theta_r - \theta_s\| = \Omega(\sqrt{|N \wmin|}) (k \log N) \sigma^2$.
\end{proof}

Let $A'$ be the matrix of points in the new space and $C'$ be the matrix of the corresponding centers. We next bound $\|A'-C'\|$. 
\begin{lemma}
\label{lem:boosting-spectral-norm}
$\|A'-C'\| \leq 24 \sigma^2 k (\sqrt{d}+2\sqrt{\log N}) \sqrt{N}$.
\end{lemma}
\begin{proof}
Let $Y = A'-C'$. Then we have that $\|Y\|^2 \leq \|Y^T Y\| = \max_{v:\|v\|=1} \sum_{r} \sum_{x \in T_r} [(x-\theta_r)\cdot v]^2$. Let $Q_r$ be the connected component in $G$ that the points in $T_r$ belong to. Then we can write $\|Y^TY\| = \max_{v:\|v\|=1} \sum_{r} \sum_{x \in T_r} \sum_{B_j \in Q_r} v^2_j [(x-\mu_r)\cdot (B_j-\mu)]^2 \leq \sum_{r} \sum_{B_j \in Q_r} v^2_j  \sum_{x \in T_r} [(x-\mu_r)\cdot (B_j-\mu)]^2$. Using \ref{eq:boost:var1}, we can bound the inner term as $ \sum_{x \in T_r} [(x-\mu_r)\cdot (B_j-\mu)]^2 \leq 4 |T_r| \|B_j-\mu\|^2 \sigma^2$.

Next notice that because of the way $G$ is constructed, points within the same connected component have distance at most $k \gamma$. Hence, $\|B_j - \mu\| \leq k \gamma$. Hence, $\|Y^TY\| \leq \sum_r \sum_{B_j \in Q_r} v^2_j 4|T_r| (k^2 \gamma^2)\sigma^2 \leq 4N k^2 \gamma^2 \sigma^2$. This gives the desired bound on $\|Y\| = \|A'-C'\|$.
\end{proof}
Combining the previous two lemmas we get that $\|\theta_r - \theta_s\| \geq \Omega(\sqrt{\frac{|N \wmin|}{d}}) \frac{\|A'-C'\|}{\sqrt{N}}$. We next run the initialization procedure from Section~\ref{sec:initialization} by projecting $A'$ onto the top $k$ subspace and running a $k$-means approximation algorithm. Let $\phi_1, \ldots \phi_k$ be the means obtained. Using Lemma~\ref{lem:initialization-weak} with $M^* = C'$, we get that for all $r$, $\|\phi_r - \theta_r\| \leq 20 \sqrt{k\alpha} \frac{\|A'-C'\|}{\sqrt{|N\wmin|}}$, where $\alpha$ is the approximation guarantee of the $k$-means approximation used. If $\Delta > c_0 \sqrt{\min\{k,d\} \log N}$, then we use a constant factor approximation algorithm~\cite{AhmadianNSW16}. If $\Delta > c_0 \sqrt{\min\{k,d\} \log k \log N}$, then we can use the simpler $k$-means++ algorithm~\cite{AV07}
\begin{proof}[Proof of Proposition~\ref{prop:initialization-strong}]
Assuming $N = \Omega(\frac{k^2d^2}{\wmin^2})$ we get that for all $r \neq s$, $\|\phi_r - \phi_s\| \geq 10\sqrt{kd}\frac{\|A'-C'\|}{\sqrt{N\wmin}}$. Let $P_1, \ldots P_k$ be the clustering of points in $A'$ obtained by using centers $\phi_1, \ldots \phi_k$. Then we have that for each $r$, $|T_r \triangle P_r| \leq \frac{N\wmin}{10\sqrt{d]}}$, since otherwise the total cost paid by the misclassified points will be more than $4k\|A'-C'\|^2$. Next we use the clustering $P_1, \ldots P_k$ to compute means for the original set of points in $A$. Let $\nu_1, \ldots \nu_k$ be the obtained means. We will show that for all $r$, $\|\nu_r - \mu_r\| \leq \tau \sigma$, where $\tau < \frac \Delta 4$.

Consider a particular partition $P_r$ that is uniquely identified with $T_r$. Let $n_{r,r}$ be the number of points that belong to both $P_r$ and $T_r$ and $\mu_{r,r}$ be the mean of those points. Similarly, let $n_{r,s}$ be the number of points that belong to $T_s$ originally but belong to $P_r$ in the current clustering, and let $\mu_{r,s}$ be their mean. Then, $\|\mu_r - \nu_r\| \leq \frac{n_{r,r}}{|P_r|}\|\mu_{r,r}-\mu_r\| + \sum_{s \neq r} \frac{n_{r,s}}{|P_r|}\|\mu_{r,s}-\mu_r\|$. We can bound $\|\mu_{r,r}-\mu_r\|$ by $O(\sigma)$ using \ref{eq:sr:mean:subset} and $\|\mu_{r,s}-\mu_s\|$ by $O(k (\sqrt{d} + 2\sqrt{\log N})$ using \ref{eq:sr:length} and the fact that points in $r$ and $s$ must belong to the same component in $G$. Combining we get the claim.
\end{proof}

\end{document}